\documentclass[letterpaper,11pt]{article}
\usepackage[top=1in, bottom=1in, left=1in, right=1in]{geometry}

\usepackage{microtype}

\usepackage{amsmath,amsthm,amssymb}
\usepackage{thm-restate}
\usepackage{mathtools}
\usepackage{xspace}
\usepackage{stmaryrd}
\usepackage{xcolor,graphicx}
\usepackage{libertine}
\usepackage[T1]{fontenc}

\usepackage{enumitem}
\setlist{nosep}

\theoremstyle{plain}
\newtheorem{theorem}{Theorem}

\newtheorem{lemma}[theorem]{Lemma}
\newtheorem{corollary}[theorem]{Corollary}

\newtheorem{proposition}[theorem]{Proposition}
\theoremstyle{definition}

\theoremstyle{remark}
\newtheorem{remark}[theorem]{Remark}

\usepackage[hidelinks]{hyperref}

\definecolor{ablue}{rgb}{0.3,0.4,0.8}
\definecolor{ared}{rgb}{0.95,0.4,0.4}
\definecolor{agreen}{rgb}{0,0.5,0.25}
\definecolor{ayellow}{rgb}{0.95,0.85,0.3}

\DeclarePairedDelimiter{\civ}{[}{]}

\newcommand{\ceil}[1]{\left\lceil #1 \right\rceil}
\newcommand{\floor}[1]{\left\lfloor #1 \right\rfloor}

\newcommand{\Reals}{\mathbb{R}}
\newcommand{\Nats}{\mathbb{N}}
\newcommand{\Ints}{\mathbb{Z}}
\newcommand{\Hyp}{\mathbb{H}}

\newcommand{\ba}{\mathbf{a}}

\newcommand{\cS}{\mathcal{S}}

\newcommand{\cP}{\mathcal{P}}

\newcommand{\cL}{\mathcal{L}}
\newcommand{\cB}{\mathcal{B}}

\newcommand{\cT}{\mathcal{T}}
\newcommand{\cM}{\mathcal{M}}
\newcommand{\cI}{\mathcal{I}}
\newcommand{\cH}{\mathcal{H}}

\newcommand{\sig}{\sigma}
\newcommand{\eps}{\varepsilon}

\newcommand{\E}{\mathbf{E}}

\newcommand\eqdef{\mathrel{\overset{\makebox[0pt]{\mbox{\normalfont\tiny\sffamily def}}}{=}}}

\DeclareMathOperator{\diam}{diam}
\DeclareMathOperator{\dist}{dist}

\newcommand{\myin}{\mathrm{in}}
\newcommand{\myout}{\mathrm{out}}
\newcommand{\bags}{\mathrm{bags}}
\newcommand{\tw}{\mathrm{tw}}

\DeclareMathOperator{\len}{length}

\newcommand{\rhoin}{\rho_\myin}
\newcommand{\rhoout}{\rho_\myout}

\newcommand{\UBG}{\mathrm{UBG}}
\newcommand{\NUBG}{\mathrm{NUBG}}
\newcommand{\SNUBG}{\mathrm{SNUBG}}
\newcommand{\Nei}{\mathrm{N}}
\newcommand{\NbG}{\mathrm{NbG}}
\newcommand{\mypath}{\mathrm{Path}}

\newcommand{\bd}{\partial}

\DeclareMathOperator{\Vol}{Vol}
\DeclareMathOperator{\poly}{poly}

\renewcommand{\leq}{\leqslant}
\renewcommand{\geq}{\geqslant}

\newcommand{\etal}{\emph{et~al.}}
\newcommand{\NP}{\mbox{\ensuremath{\mathsf{NP}}}\xspace}

\newcommand{\claim}[2]
{\vspace{2\parskip}\noindent {\emph{Claim.} }{#1} \\[\parskip]
{\emph{Proof of claim.}} #2 \hfill {\footnotesize $\Box$}}

\newcommand{\IS}{\textsc{Independent Set}\xspace}
\newcommand{\GT}{\textsc{Grid Tiling}\xspace}
\newcommand{\GTn}{\textsc{Grid Tiling}}
\newcommand{\GTleq}{\textsc{Grid Tiling with $\leq$}\xspace}

\newcommand{\niceSAT}{($3,3$)-\textsc{SAT}\xspace}
\newcommand{\niceCNF}{($3,3$)-CNF\xspace}

\newcommand{\obemail}{\raisebox{-0.21em}{\protect\includegraphics[height=0.8em]{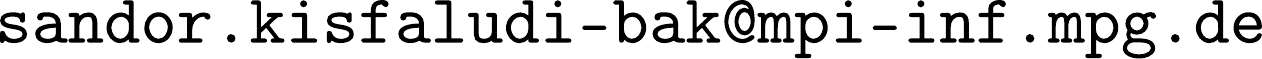}}}

\DeclarePairedDelimiter{\norm}{\lVert}{\rVert}

\makeatletter
\DeclareRobustCommand{\bfseries}{%
  \not@math@alphabet\bfseries\mathbf
  \fontseries\bfdefault\selectfont
  \boldmath
}
\makeatother

\title{Hyperbolic intersection graphs and (quasi)-polynomial time}
\date{}

\author{S\'andor Kisfaludi-Bak\thanks{Max Planck Institut f\"ur Informatik, Saarbr\"ucken, Germany; email: \obemail{}. This work was supported by the Netherlands Organization for Scientific Research NWO under project no. 024.002.003.}}

\begin{document}

\begin{titlepage}
\maketitle
\thispagestyle{empty}
\begin{abstract}
We study unit ball graphs (and, more generally, so-called noisy  uniform ball graphs) in $d$-dimensional hyperbolic space, which we denote by $\Hyp^d$. Using a new separator theorem, we show that unit ball graphs in $\Hyp^d$ enjoy similar properties as their Euclidean counterparts, but in one dimension lower: many standard graph problems, such as \IS, \textsc{Dominating Set, Steiner Tree}, and \textsc{Hamiltonian Cycle} can be solved  in $2^{O(n^{1-1/(d-1)})}$ time for any fixed $d\geq 3$, while the same problems need $2^{O(n^{1-1/d})}$ time in $\Reals^d$. We also show that these algorithms in $\Hyp^d$ are optimal up to constant factors in the exponent under ETH.

This drop in dimension has the largest impact in $\Hyp^2$, where we introduce a new technique to bound the treewidth of noisy uniform disk graphs. The bounds yield quasi-polynomial ($n^{O(\log n)}$) algorithms for all of the studied problems, while in the case of \textsc{Hamiltonian Cycle} and \textsc{$3$-Coloring} we even get polynomial time algorithms. Furthermore, if the underlying noisy disks in $\Hyp^2$ have constant maximum degree, then all studied problems can be solved in polynomial time. This contrasts with the fact that these problems require $2^{\Omega(\sqrt{n})}$ time under ETH in constant maximum degree Euclidean unit disk graphs.

Finally, we complement our quasi-polynomial algorithm for \IS in noisy uniform disk graphs with a matching $n^{\Omega(\log n)}$ lower bound under ETH. This shows that the hyperbolic plane is a potential source of \NP-intermediate problems.
\end{abstract}
\end{titlepage}


\section{Introduction}

Hyperbolic space has seen increasing interest in recent years from various communities in computer science due to its unique metric properties. It has been connected to several topics, for example random networks~\cite{FriedrichK18,BlasiusFK16,GugelmannPP12}, routing and load balancing~\cite{Kleinberg07,ShavittT08}, metric embeddings~\cite{Sarkar12,VerbeekS16}, and visualization~\cite{MunznerB95}. The algorithmic properties of hyperbolic space have been mostly studied through Gromov's hyperbolicty, which is a convenient combinatorial description for negatively curved spaces~\cite{KrauthgamerL06,ChepoiE07,FournierIV15}.


In this paper, we study treewidth bounds for intersection graphs of unit-ball-like objects in hyperbolic space. These intersection graphs are a natural choice to capture some important properties of the underlying metric. The most well-studied geometric intersection graphs are unit disk graphs in $\Reals^2$. From the perspective of treewidth and exact algorithms, unit disk graphs have some intriguing properties: they are potentially dense, (may have treewidth $\Omega(n)$), but they still exhibit the ``square root phenomenon'' for several problems just as planar and minor-free graphs do; so for example one can solve \IS or \textsc{$3$-coloring} in these classes in $2^{O(\sqrt{n})}$ time~\cite{Marx13}, while these problems would require $2^{\Theta(n)}$ time in general graphs unless the Exponential Time Hypothesis (ETH)~\cite{ImpagliazzoP01} fails. In $\Reals^d$, the best \IS running time for unit ball graphs is $2^{\Theta(n^{1-1/d})}$~\cite{MarxS14,frameworkpaper}.
Note that $d$-dimensional Euclidean space has bounded doubling dimension, or in other words, Euclidean space has polynomial growth: balls of radius $r$ have volume $\poly(r)$. It turns out that polynomial growth by itself can yield subexponential algorithms, for example for the \textsc{Steiner Tree} problem~\cite{DBLP:conf/esa/MarxP17}. One of the key distinguishing features of hyperbolic space is that it has exponential growth, which requires a new approach.

To our knowledge, the only paper studying treewidth of graphs in hyperbolic space is the work by Bl{\"{a}}sius, Friedrich, and Krohmer~\cite{BlasiusFK16}, who investigate random hyperbolic disk graphs, where the disks have equal radii, and the disk centers are chosen from some distribution. They prove various treewidth bounds depending on the parameter of the distribution. In a new manuscript, Bl{\"{a}}sius~\etal~show that the problem of \textsc{Vertex Cover} can be solved in polynomial time on hyperbolic random graphs with high probability~\cite{DBLP:journals/corr/abs-1904-12503}.
Our goal is to get worst case bounds on the treewidth of intersection graphs.
Naturally, getting a sublinear bound on separators or
treewidth itself is not possible, since cliques are unit ball graphs.
Therefore, we use the partition and weighting scheme developed by De
Berg~\etal~\cite{frameworkpaper}. Given an intersection graph $G=(V,E)$, one
defines a partition $\cP$ of the vertex set $V$; initially, it is useful to
think of a partition into cliques using a tiling of the underlying space where
each tile has a small diameter. Then the partition classes are defined for
each non-empty tile as the set of balls whose center falls inside the tile.
This gives a clique partition $\cP$ of $G$. Next, each clique $C\in \cP$ receives a
weight of $\log(|C|+1)$. It turns out that this weighting is useful for many
problems, and it motivates us to define weighted separators of $G$ with
respect to $\cP$ as a separator consisting of classes of $\cP$, whose weight
is defined as the sum of the weights of the constituent partition classes. It
is in this sense that we can find sublinear separators.

Such weightings can also be used along with treewidth techniques. Let $G$ be a graph and let $\cP$ be a partition of $V(G)$. Let $G_\cP$ be the graph obtained by contracting all edges of $G$ that go within a partition class of $\cP$, i.e., $V(G_\cP)$ can be identified with $\cP$. For each class $C\in \cP$, we assign the weight $\log(|C|+1)$. We define the \emph{$\cP$-flattened treewidth} of $G$ as the 
weighted treewidth of $G_\cP$ under this weighting. (We give more thorough definitions in Section~\ref{subsec:twbasic}.) For example, unit ball graphs in $\Reals^d$ have a clique partition $\cP$ such that their $\cP$-flattened treewidth is $O(n^{1-1/d})$~\cite{frameworkpaper}.

We intend to show that unit disk graphs in hyperbolic space are even more intriguing than the ones in Euclidean space from the perspective of computational complexity. Let $\Hyp^d$ denote $d$-dimensional hyperbolic space of sectional curvature $-1$. In hyperbolic space, the radius of the disks or balls matters. For example one gets different graph classes for radius 1 and radius 2 disks in $\Hyp^2$. Hence, we parameterize the graph class of equal-sized balls by the radius $\rho$ of these balls, and denote the class by $\UBG_{\Hyp^d}(\rho)$. 

There have been several papers studying \IS, \textsc{Dominating Set, Hamiltonian Cycle, $q$-Coloring}, etc. in unit ball graphs in Euclidean space~\cite{ito-hamiltonian,ciac-hamiltonian,MarxS14,BiroBMMR17,FominLPSZ17,frameworkpaper}, all concluding that $2^{O(n^{1-1/d})}$ is the optimal running time for these problems in $\Reals^d$. In this paper, we show that a similar phenomenon occurs in $\Hyp^d$, but shifted by one dimension: the problems can be solved in $2^{O(\sqrt{n})}$ time in $\Hyp^3$, just as in $\Reals^2$; in general, for constant $d\ge 3$, the optimal running time is $2^{\Theta(n^{1-1/(d-1)})}$ for these problems in $\Hyp^d$ under ETH. This is perhaps less surprising than it seems at first sight because of the fact that one can embed $\Reals^{d-1}$ into $\Hyp^d$ isometrically\cite{benedetti2012lectures}. Indeed, we use this connection to establish lower bounds. On the other hand, this embedding does not facilitate getting algorithms in $\Hyp^d$, and it is also not useful for establishing algorithms and lower bounds in $\Hyp^2$. In $\Hyp^d$, we give a new separator theorem, which can be regarded as the hyperbolic analogue of various separator theorems known for (similarly sized) fat objects in Euclidean space~\cite{MillerTTV97,SmithW98,AlberF04,frameworkpaper}.\footnote{Several Euclidean separation techniques fail in $\Hyp^d$ since using compact hypersurfaces as separator objects (i.e., a sphere or the boundary of a hypercube) cannot produce small balanced separators due to the linear isoperimetric inequality.}

In $\Hyp^2$, we use a novel approach to gain tight treewidth bounds: our
argument builds on the isoperimetric
inequality~\cite{teufel1991generalization} and on the treewidth bound of
$k$-outerplanar graphs~\cite{Bodlaender98}. As a consequence, we see large
drop in complexity compared to $\Reals^2$: several problems that are
\NP-complete in unit disk graphs in $\Reals^2$ can be solved in
quasi-polynomial ($n^{O(\log n)}$) time for uniform disk graphs of constant
radius in $\Hyp^2$, or even in polynomial time if the graphs additionally have
constant maximum degree. Moreover, we show that the
quasi-polynomial running time in the general case is likely unavoidable by giving a
quasi-polynomial lower bound for \IS; this is yet another difference from the
Euclidean setting where the problem has the same $2^{\Theta(\sqrt{n})}$
running time both generally and in case of constant maximum degree. This is perhaps the most striking consequence of our
treewidth bounds. We also identify two problems, namely \textsc{$q$-Coloring}
(for constant $q$) and \textsc{Hamiltonian Cycle} that admit polynomial time
algorithms even in case of unbounded degree.  
Quasi-polynomial exact algorithms with matching lower bounds are rare,
and so far we know of only a few natural problems in this class. Examples
include \textsc{VC-dimension} and \textsc{Tournament Dominating
Set}~\cite{DBLP:journals/tcs/MegiddoV88,%
DBLP:journals/iandc/LinialMR91,DBLP:journals/jcss/PapadimitriouY96}, and a
weighted coloring problem on trees\cite{DBLP:journals/siamdm/AraujoNP14}. It
seems that problems on intersection graphs in the hyperbolic plane are a
natural source of further problems in this class.



Before we can state our contributions precisely, we first define the graph classes that
we consider more formally. Let $\cM$ be a
metric space with distance function $\dist_\cM:\cM\times\cM \rightarrow \Reals_{\geq0}$, and let $\rho>0$ and $\nu\geq 1$ be real numbers. A  graph
$G=(V,E)$ is a radius $\rho$ uniform ball graph with noise $\nu$ (or a
\emph{$(\rho,\nu)$-noisy uniform ball graph} for short) if there is a
function $\eta:V\rightarrow \cM$, such that for all pairs $v,w\in V$ we have:
\begin{itemize}
\item $\dist_\cM(\eta(v),\eta(w)) < 2\rho \Rightarrow (v,w)\in E$
\item $\dist_\cM(\eta(v),\eta(w)) > 2\rho\nu \Rightarrow (v,w)\not\in E$.
\end{itemize}
In particular, pairs of vertices $v,w$ where $\dist_\cM(\eta(v),\eta(w)) \in
[2\rho,2\rho\nu]$ can either be connected or disconnected. We denote the class
of graphs defined this way by $\NUBG_\cM(\rho,\nu)$.

This is a generalization of unit ball graphs: clearly $\UBG_\cM(1) \subseteq
\NUBG_\cM(1,1)$. Although the class $\NUBG_{\Hyp^d}(\rho,\nu)$ seems like a slight
generalization of $\UBG_{\Hyp^d}(\rho)$, it corresponds to a much larger graph class;
this is shown in Theorem~\ref{thm:noisybig} in Appendix~\ref{sec:app_nubg}.

The function $\eta$ is called an \emph{embedding} of $G$. Note that $\eta$ is
not necessarily injective, so its image is a multiset. An intersection graph of similarly sized fat objects with inscribed and
circumscribed ball radii $\rhoin$ and $\rhoout$ is a noisy unit ball graph
with $\rho=\rhoin$ and $\nu=\rhoout/\rhoin$. Therefore noisy uniform ball
graphs are a direct generalization of intersection graphs of similarly sized
fat objects as seen in~\cite{frameworkpaper}.

For all of our results, we require that $\nu$ is a fixed constant. It is necessary to bound $\nu$ in some way, since for $\nu=n$, or in case of hyperbolic space even for
$\nu=c\log n$ the class $\NUBG_{\Hyp^d}(\rho,\nu)$ includes all graphs on $n$
vertices. Changing the parameter $\rho$ also has important effects.
\footnote{Alternatively, one could also think of fixing $\rho=1$ and changing the
sectional curvature of the underlying space.} For very small values of $\rho$,
the resulting graph class is very similar to the set of Euclidean
unit disk graphs. In order to make sure that all of our techniques work, we are forced to
fix $\rho$ to be a positive absolute constant.

\begin{remark}
It would be worthwhile to investigate the case where $\rho$ may depend on $n$. In our model, for any fixed $\nu>1$ there is a radius $\rho=
\rho(n,\nu)$ such that $\NUBG(\rho,\nu)$ contains all graphs on $n$ vertices.
So to get interesting treewidth bounds, one should first try to resolve
the case $\nu=1$. The resulting graph classes would also include hyperbolic random graphs. We leave this for future work.
\end{remark}

We are also interested in noisy uniform ball graphs that are shallow, as defined next. A point set $P$ is $(\rho,k)$-\emph{shallow} if any ball of radius $\rho$ can contain at most $k$ points from $P$. We denote by $\SNUBG_\cM(\rho,\nu,k)$ the set of graphs in $\NUBG_\cM(\rho, \nu)$ that have a $(\rho,k)$-shallow embedding. For any fixed $\rho, \nu$ and $k$, the shallow class $\SNUBG_\cM(\rho,\nu,k)$ is a very small subset of $\NUBG_\cM(\rho, \nu)$ by part (ii) and (iii) of Theorem~\ref{thm:noisybig}. 

\vspace{-0.3cm}
\paragraph*{Our contribution.} 
Our first contribution is a separator theorem that is inspired by the Euclidean weighted separator of Fu~\cite{Fu2011}. Our separator for a graph $G$ is based on a partition $\cP$ of $V(G)$ where each class $C\in \cP$ induces a clique in $G$. The separator itself will be a set $S\subset \cP$.  The weight of a separator $S$ is defined as $\gamma(S)\eqdef\sum_{C\in S} \log(|C|+1)$. We say that $S$ is \emph{$\beta$-balanced} if all connected components induced by  $V(G)\setminus (\bigcup_{C\in S} C)$ have at most $\beta n$ vertices.

\begin{theorem}\label{thm:sep}
Let $d\geq 2$, $\rho>0$ and $\nu\geq 1$ be constants, and let $G\in \NUBG_{\Hyp^d}(\rho,\nu)$.
Then $G$ has a clique partition and a corresponding $\frac{d}{d+1}$-balanced
clique-weighted separator $S$ such that
\begin{enumerate}
\item[(i)] if $d=2$, then $S$ has size $|S|=O(\log n)$ and weight $\gamma(S)=O(\log^2 n)$
\item[(ii)] if $d\geq 3$, then $S$ has size $|S|=O(n^{1-1/(d-1)})$ and weight $\gamma(S)=O(n^{1-1/(d-1)})$.
\end{enumerate}
Given $G$ and an embedding into $\Hyp^d$, the clique partition and the separator can be found with a Las Vegas algorithm in $\poly(n)$ expected time. 
\end{theorem}

\begin{table}[t]
\centering
  \begin{tabular}{l | l l l l}
  & $\SNUBG_{\Hyp^2}$ & $\NUBG_{\Hyp^2}$ & $\NUBG_{\Hyp^d},\, d\geq 3$\\
  \hline
  \IS \& others & $\poly(n)$ & $n^{O(\log n)}$ & $2^{O(n^{1-1/(d-1)})}$ \\
  \textsc{$q$-Coloring}, $q=O(1)$ & $\poly(n)$ & $\poly(n)$ & $2^{O(n^{1-1/(d-1)})}$ \\
  \textsc{Hamiltonian Cycle} & $\poly(n)$ & $\poly(n)$ {\scriptsize(emb)} & $2^{O(n^{1-1/(d-1)})}$ {\scriptsize(emb)}\\
  \hline
  \end{tabular}
\caption{Summary of algorithmic consequences. The first row contains the running time for the following problems: \IS, \textsc{Dominating Set, Vertex Cover, Feedback Vertex Set, Connected Dominating Set, Connected Vertex Cover, Connected Feedback Vertex Set}. Algorithms with the note ``{\scriptsize(emb)}'' require an embedding as part of the input.}
\label{tab:results}
\end{table}

In case of $d\geq 3$,  we can apply a standard way of turning balanced
separators into a bound on treewidth, yielding tight treewidth bounds. On the other hand, for $d=2$
the size of the separator is logarithmic, and one gets an extra
logarithmic factor in the treewidth bound. This is significant because the
treewidth typically shows up in the exponent of the running time, which means
that the extra logarithmic factor can make the difference between polynomial
and quasi-polynomial time. Our second contribution shows that the extra logarithmic factor in the
treewidth can be avoided. Note that part (i) of the theorem is about shallow noisy uniform disk graphs, while part (ii) is about general noisy
uniform disk graphs.

\begin{theorem}\label{thm:treewidth}
Let $\rho>0$, $\nu\geq 1$ and $k\in \Nats_+$ be constants.
\begin{enumerate}
\item[(i)] If $G\in \SNUBG_{\Hyp^2}(\rho,\nu,k)$, then the
treewidth of $G$ is $O(\log n)$. Given $G$, a tree decomposition of width $O(\log n)$ can be found in polynomial time.
\item[(ii)] If $G\in \NUBG_{\Hyp^2}(\rho,\nu)$ , then there exists a
clique partition $\cP$ such that the $\cP$-flattened treewidth of $\,G$ is
$O(\log^2 n)$. Given $G$ and an embedding into $\Hyp^2$, a clique partition $\cP$ and a corresponding decomposition of weight $O(\log^2 n)$ can be found in polynomial time.
\end{enumerate}
\end{theorem}

In order to find tree decompositions in the case $d\geq 3$, we can use techniques of De
Berg~\etal~\cite{frameworkpaper}, as detailed in Section~\ref{sec:gettw}
to make our results for the case $d\geq 3$ slightly more practical from the perspective of algorithms. We show
that tree decompositions of width asymptotically matching our bounds can be
found in these graph classes if an embedding $\eta$ is given. This approach
avoids the geometric separator search of Theorem~\ref{thm:sep}, and it is
deterministic. %
%
In the absence of a geometric embedding, the techniques
of~\cite{frameworkpaper} can be used to compute tree decompositions, under the
promise that $G$ is of some particular graph class
$\SNUBG_{\Hyp^d}(\rho,\nu,k)$ or $\NUBG_{\Hyp^d}(\rho,\nu)$. For
$\SNUBG_{\Hyp^d}(\rho,\nu,k)$, a tree decomposition of
width at most constant times our treewidth bounds can be computed for any
$d\geq 2$. If our graph is non-shallow, then we get weaker tree decompositions: we can compute a so-called \emph{greedy partition}
$\cP$ that mimics a clique partition, and a corresponding tree decomposition
of width asymptotically matching our bounds.


As our final contribution, we complement the quasi-polynomial algorithm for \IS in $\Hyp^2$ with a matching lower bound.\footnote{The lower bound is for a specific radius, but is easily adaptable to other constant radii.}

\begin{restatable}{theorem}{thmlowerbound}\label{thm:lowerbound}
For $\rho=\cosh^{-1}(\cos(\pi/5)/\sin(\pi/4))$, there is no $n^{o(\log n)}$ algorithm for \IS in $\NUBG_{\Hyp^2}(\rho,1)$ unless ETH fails.
\end{restatable}

We show that the lower bound framework of~\cite{frameworkpaper} for $\Reals^d$ applies in $\Hyp^{d+1}$. Consequently, all of our algorithms given in $\Hyp^d, d\ge 3$ have matching lower bounds under ETH.

\section{Tilings and a separator in \texorpdfstring{$\Hyp^d$}{hyperbolic space}}\label{sec:tilings}

The goal of this section is to prove Theorem~\ref{thm:sep}. The clique partition will be induced by a tiling, and the separator $S$ will be defined as the collection of tiles intersected by the neighborhood of a random hyperplane. This requires that we introduce customized tilings of $\Hyp^d$, and study how tiles are intersected by random hyperplanes first.

\subsection{Custom tilings and their properties}

A \emph{tiling} $\cT$ of $\Hyp^d$ is a set of interior disjoint compact subsets of $\Hyp^d$ that together cover $\Hyp^d$; in this paper, we only use tiles that are homeomorphic to a closed  ball. We define tilings where each tile contains a ball of radius $\rho$, and each tile is contained in a ball of radius $\rho\nu$ for some absolute constants $\rho> 0$ and $\nu\geq 1$. We call such a tiling a $(\rho,\nu)$-\emph{nice tiling}, or just nice tiling for short. A pair of distinct tiles $T,T'\in \cT$ are \emph{neighboring} if their boundaries intersect in more than one point, and they are \emph{isometric} if there is a distance-preserving transformation of $\Hyp^d$ that maps $T$ to $T'$. A tiling of $\Hyp^2$ with bounded tiles is regular if it is vertex-, edge- and face-transitive, or equivalently, if it is a tiling with copies of a fixed bounded regular polygon where each vertex of the tiling is a vertex of all the containing polygons.\footnote{We have excluded ideal polygons, although they do define regular tilings.} Such regular tilings are always nice, since each face is a copy of a regular polygon that has an inscribed disk of positive radius and a finite radius circumscribed disk.

\begin{restatable}{lemma}{lemtilings}\label{lem:tilings}
\textcolor{white}{.}
\begin{enumerate}
	\item[(i)] There are positive constants $c_1$ and $c_2$ such that for any $\delta\in \Nats_+$, there is a regular tiling $\cT_\delta$ of $\,\Hyp^2$ where non-neighboring tiles have distance at least $c_1 \delta$, the tiles have diameter $c_2\delta$, and the ratio of their area $\alpha_\delta$ and their perimeter $\xi_\delta$ converges and satisfies $\lim_{\delta \rightarrow \infty} \frac{\alpha_\delta}{\xi_\delta} >0$.
	\item[(ii)] For any $\delta\in \Reals_+$ there is a nice tiling of $\,\Hyp^d$ with isometric tiles of diameter $2\delta$.
\end{enumerate}
\end{restatable}

\begin{proof}

(i)
Consider the tiling $\cT$ of $\Hyp^2$ by regular $2^{\delta+2}$-gons where at
each vertex three $2^{\delta+2}$-gons meet (that is, the regular tiling with
Schl\"afli symbol $(2^{\delta+2},3)$). We examine the right triangle created
by the center $o$ of a tile, a midpoint $p$ of a side and a neighboring vertex
$v$, see Figure~\ref{fig:tiles83}. The triangle has angle $\pi/3$ at $v$,
since there are six such angles at $v$, and it has angle $\pi/2^{\delta+2}$ at
$o$. From the hyperbolic law of cosines for angles ($\cos C = -\cos A \cos B +
\sin A\sin B \cosh c$) we have that the hypotenuse $ov$ has length

\begin{figure}[t]
\centering
\includegraphics[width=0.55\textwidth,trim={5cm 5.5cm 0cm 2.5cm},clip]{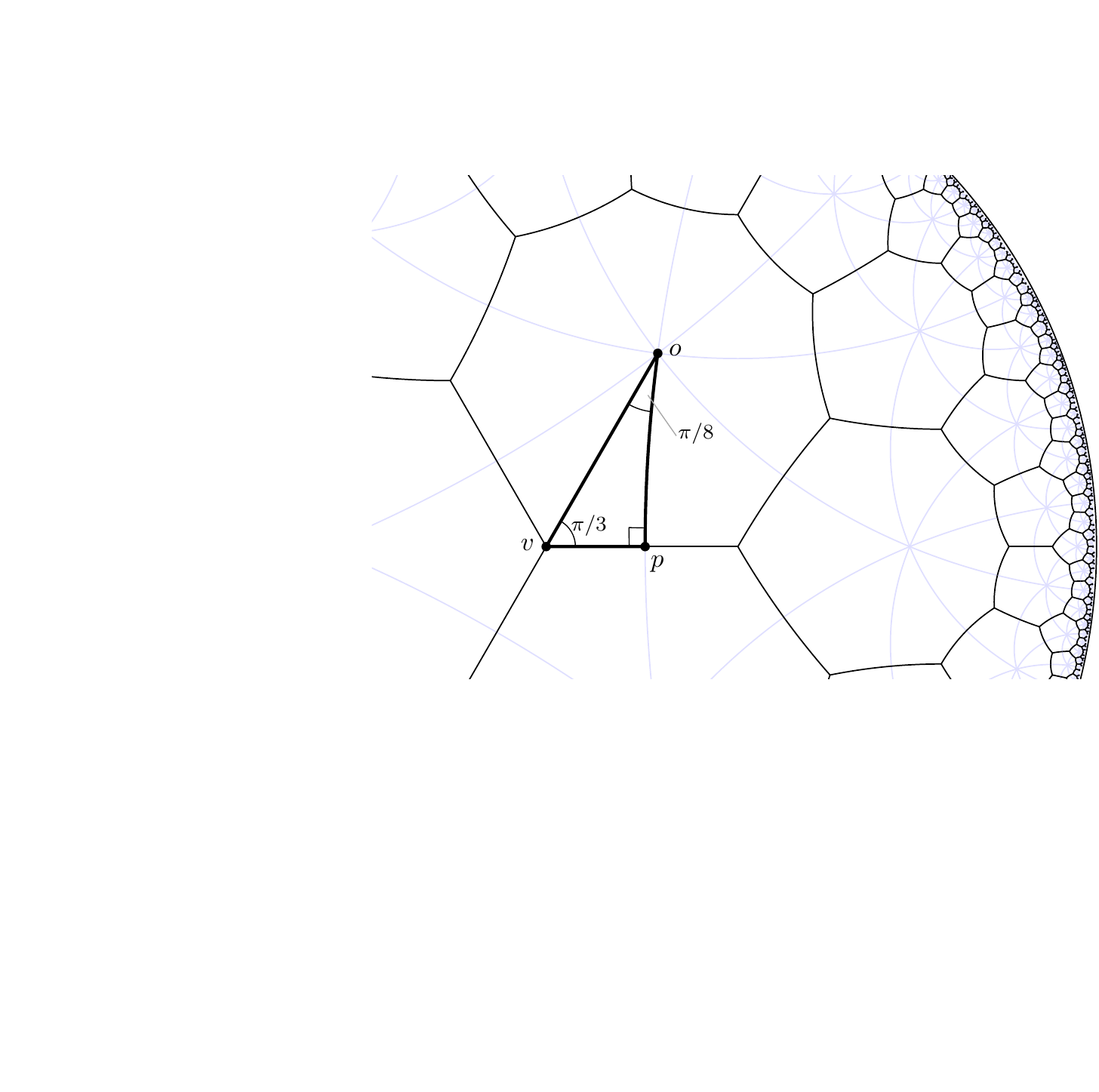}
\caption{A regular tiling with octagons ($\delta=1$).}\label{fig:tiles83}
\end{figure}

\[|ov| = \cosh^{-1}\left( \cot(\pi/3)\cot(\pi/2^{\delta+2}) \right) = \cosh^{-1}(\Theta(2^{\delta+2}/\pi)) = \Theta(\delta),\]

where the second step uses that $\cot(x) = 1/x -O(x)$, which follows from the 
Laurent-series of $\cot(x)$ around $0$. Consequently, the diameter of the
tiles is $2|ov|=\Theta(\delta)$. It is useful to find the lengths of the other
two sides as well. 
\begin{align*}|vp| &=
\cosh^{-1}(\cos(\pi/2^{\delta+2})/\sin(\pi/3)) = \Theta(1)\text{ and}\\
|po| &=
\cosh^{-1}(\cos(\pi/3)/\sin(\pi/2^{\delta+2}))= \Theta(\delta)
\end{align*}
The distance
of non-neighboring tiles is at least $2|vp|=\Theta(1)$.

The area of each tile is $2^{\delta+3}$ times the area of the triangle $pov$,
so $\alpha_\delta=2^{\delta+3}(\pi-(\frac{\pi}{2} + \frac{\pi}{3} +
\frac{\pi}{2^{\delta+2}}))$. The circumference of each tile is
$\xi_\delta=2^{\delta+3}|vp|$. Therefore we have that the limit of the ratio of
the area and the circumference is

\[\lim_{\delta\rightarrow \infty}\,\frac{2^{\delta+3}(\pi-(\frac{\pi}{2}
+ \frac{\pi}{3} + \frac{\pi}{2^{\delta+2}}))}%
{2^{\delta+3}\cosh^{-1}(\cos(\frac{\pi}{2^{\delta+2}})/\sin(\frac{\pi}{3}))}
=\frac{\pi/6}{\cosh^{-1}(\cos 0 /\sin(\frac{\pi}{3}))}>0.953.\]

\begin{figure}[t]
\centering
\includegraphics[height=5cm]{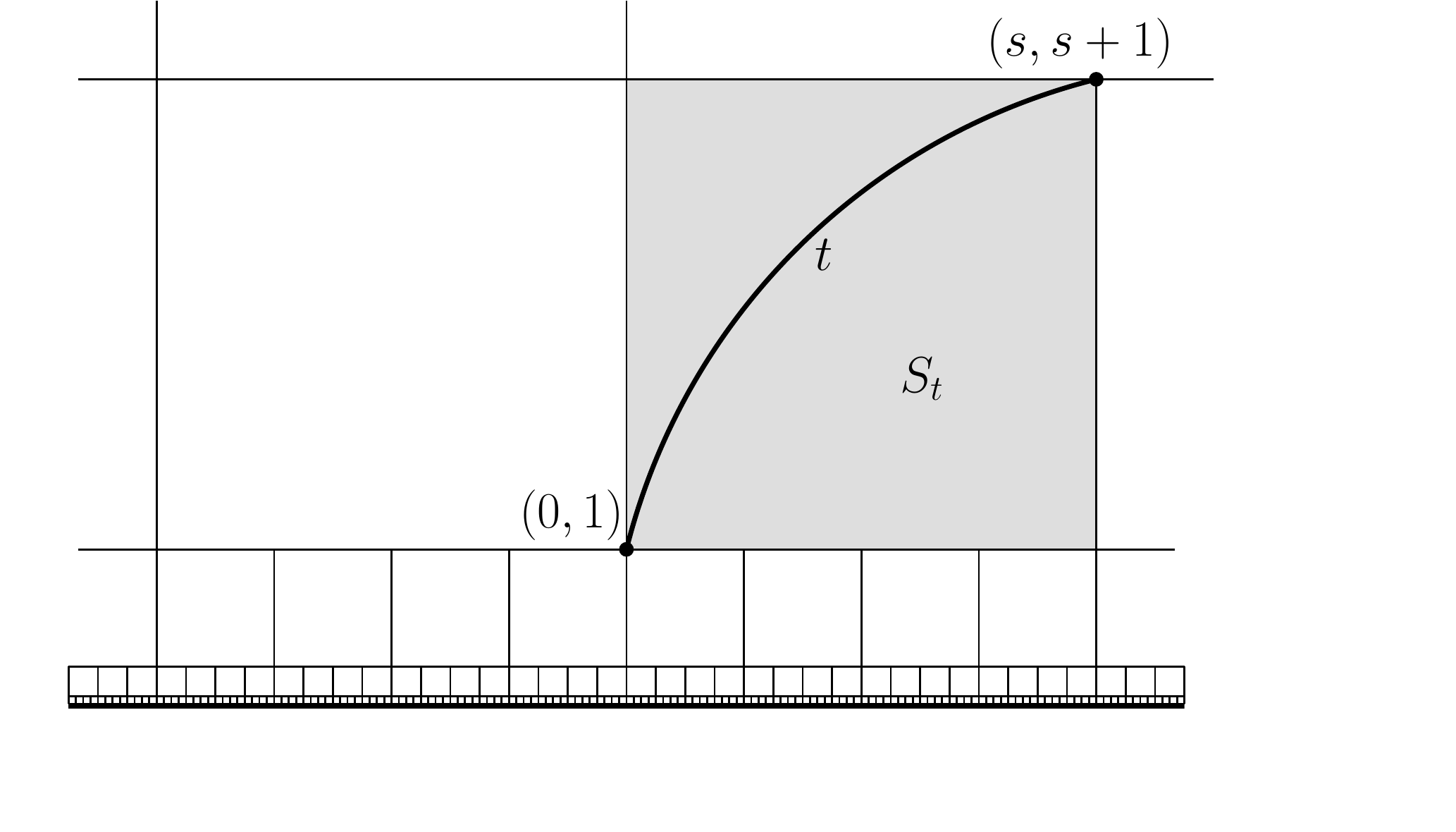}
\caption{A tiling of $\Hyp^2$ where the base Euclidean square $S_t$ has (Euclidean) side length $s=3$.}\label{fig:t_tiling}
\end{figure}

(ii) We describe the tiles using the Poincar\'e half-space model of $\Hyp^d$. The tiling is based on the tiling of Cannon~\etal~\cite[Section~14]{CannonFKP97}. Let $S_{\delta}$ be the axis-aligned Euclidean hypercube with lexicographically minimal corner $(0,\dots,0,1)$ and hyperbolic diameter $\delta$; let $s$ be its Euclidean side length. (Note that $S_\delta$ does not resemble a hypercube in any sense in $\Hyp^d$.) See Figure~\ref{fig:t_tiling}. The hyperbolic diameter of $S_\delta$ is realized for two opposing vertices of the hypercube, e.g. for the vertices $(0,\dots,0,1)$ and $(s,\dots,s,s+1)$, so its diameter is
\[\delta=\cosh^{-1}\left(1+ \frac{d\cdot s^2}{2(s+1)}  \right),\]
which gives
\[s=\frac{\cosh(\delta)-1+\sqrt{(\cosh(\delta)-1)^2+2d(\cosh(\delta)-1)}}{d}.\]
We note that for $\delta>1$, we have $s=\Theta(\cosh(\delta)-1)=\Theta(e^\delta)$, and for $\delta\leq 1$ we have $s=\Theta(\sqrt{\cosh(\delta)-1}) = \Theta(\delta)$.
Then we can define $\cT$ as the image of $S_\delta$ for the following set of Euclidean transformations:
\[f_{\ba,b}(p)= (s+1)^b\big(p+(s\ba,0)\big),\]
where $\ba\in \Ints^{d-1} \times \{0\}$ and $b\in \Ints$. The resulting tiling is illustrated in Figure~\ref{fig:t_tiling}: we essentially take horizontal translates of $S_\delta$ to tile a strip, then scale it so that a thinner strip below it is covered, then scale it again to cover an even thinner strip below that, etc; we similarly cover thicker strips above $S_\delta$. Note that $f_{\ba,b}$ is a hyperbolic isometry since it can be regarded as the succession of a translation and a homothety from the origin, both of which are isometries of the half-space model~\cite{CannonFKP97}.
Finally, a Euclidean hypercube has a Euclidean inscribed ball, which is also an inscribed ball in the hyperbolic sense as well, although with different center and radius. A straightforward calculation gives that the hyperbolic radius of this ball is $\frac{1}{2}\ln(s+1)$.
In case of $\delta>1$, the radius is $\frac{1}{2}\ln(s+1) = \ln(\Theta(e^\delta))=\Theta(\delta)$, and when $\delta \leq 1$, the radius is $\frac{1}{2}\ln(s+1) =\ln(1+\Theta(\delta))=\Theta(\delta)$.
On the other hand, the tiles have diameter $\delta$, so they have a circumscribed ball of radius $\delta$. Hence, the tiling is nice.
\end{proof}

Given a metric space $\cM=(X,\dist)$ and a subset $S\subset X$, the \emph{$\rho$-neighborhood}\index{r@$\rho$-neighborhood} of $S$ is $\Nei_\cM(S,\rho)\eqdef \{x\in X \;|\; \inf_{s\in S} \dist(x,s)\leq \rho\}$. Given a tiling $\cT$ and a set of tiles $S\subseteq \cT$, the \emph{neighborhood graph}\index{graph!neighborhood \tilde} of $S$ is the graph with vertex set $S$ where tiles are connected if and only if they are neighbors. We denote the neighborhood graph by $\NbG(S)$, and its shortest-path distance by $\dist_S\eqdef\dist_{\NbG(S)}$. We slightly abuse notation and refer to the metric space $(\cT,\dist_\cT)$ as ``$\cT$''.

Since non-neighboring tiles of the tiling in Lemma~\ref{lem:tilings}(i) are at distance at least $c_\delta=\Omega(1)$, a pair of tiles at $\Hyp^2$-distance smaller than $c_\delta$ defines an edge of $\NbG(\cT_\delta)$. We will need the following corollary in the next section.

\begin{corollary}\label{cor:nb_tiles}
Let $\delta\in \Nats_+$, and let $\cT$ be a tiling with tiles of diameter $\Theta(\delta)$ as defined by part (i) of Lemma~\ref{lem:tilings}. Then there exists a constant $c_\delta>0$ such that for any set of tiles $S\subset\cT$, the tiles in $S$ and their neighboring tiles cover the $c_\delta$-neighborhood of $S$, that is,
$\Nei_{\Hyp^2}\left(\bigcup S,c_\delta\right) \subseteq \bigcup \Nei_\cT(S,1).$
\end{corollary}

Our results strongly rely on some properties of $\Hyp^d$ that are consequences of its exponential expansion. 

\begin{restatable}{proposition}{proptilesballs}\label{prop:tilesballs}
Let $d\geq 2$, $\rho>0$, and $\nu\geq 1$ be constants. Then for any $(\rho,\nu)$-nice tiling $\cT$ of $\Hyp^d$ we have:
\begin{enumerate}
	\item[(i)] The number of tiles from $\cT$ intersected by a radius $r$ ball is $\Theta(e^{(d-1)r})$.
	\item[(ii)] Let $\tau=2\rho\nu$, and for a positive integer $j$ let $\cL_j\subset \cT$ be the set of tiles whose distance from the origin is in the interval $\left[(j-1)\tau,j\tau\right)$. Then $|\cL_j|=\Theta(e^{(d-1)j\tau})$.
\end{enumerate}
\end{restatable}

\begin{proof}
\noindent (i) The volume of a ball of radius $r$ is
$\Vol_d(r)= \omega_d\int_{0}^r \sinh^{d-1}(x)\, \mathrm{d}x$,
where $\omega_d$ depends only on $d$, hence $\Vol_d(r)= \Theta(e^{(d-1)r})$~\cite{Rivin09}.
We denote by $B(x,r)$ the ball in $\Hyp^d$ with center $x$ and radius $r$, and let $o$ be the origin.
The 
tiles intersected by $B(o,r)$ together cover $B(o,r)$. Therefore, the number of tiles contained in $B(o,r)$ is at least $\Vol_d(r)/\Vol_d(\rho\nu) = \Omega(e^{(d-1)r})$. For the upper bound, notice that inscribed balls of tiles are interior disjoint, and that each tile has diameter at most $\tau=2\rho\nu$. Therefore tiles intersecting $B(o,r)$ are contained in $B(o,r+2\rho\nu)$, so there are at most $\Vol_d(r+2\rho\nu)/\Vol_d(\rho) = O(e^{(d-1)r})$ tiles contained in $B(o,r)$.

\medskip
\noindent (ii) Clearly all tiles of $\cL_j$ intersect  $B(o,j\tau)$, therefore the upper bound of (i) carries over. For a lower bound, let $\cB$ be the set of balls of radius $\rho\nu$ drawn around the tiles of $\cL_j$. Since any tile that intersects the sphere $\bd B(o,j\tau)$ and the interior of $B(o,j\tau)$ must be in $\cL_j$, the union $\bigcup \cB$ covers the sphere $\bd B(o,r)$. Each ball in $\cB$ can cover at most $O(1)$ total (surface) volume of this sphere, whose total (surface) volume is
$\Vol(S^{d-1}_\Hyp (j\tau))= \omega_d \sinh^{d-1}(j\tau)= \Theta(e^{(d-1)j\tau})$.
Consequently, $|\cL_r| = \Omega(e^{(d-1)j\tau})$.
\end{proof}

We can define a \emph{uniform random hyperplane} of $\Hyp^d$ through $o$ the following way. Consider the Poincar\'e ball model with center $o$, and let $H$ be a hyperplane in $\Reals^d$ through $o$ whose normal is chosen uniformly at random from the unit sphere. The intersection of $H$ with the open unit ball of the Poincar\'e model is a uniform random hyperplane of $\Hyp^d$ through $o$.

\begin{restatable}{lemma}{lemintersectingtile}\label{lem:intersectingtile}
Let $d\geq 2$, $\rho>0$, $\nu\geq 1$, and $\delta>0$ be constants. Fix a $(\rho,\nu)$-nice tiling $\cT$ of $\Hyp^d$, and set $\tau=2\rho\nu$. Let $T\in \cT$ be a tile whose distance from the origin is in the interval $\left[(j-1)\tau,j\tau\right)$ (i.e., $T\in \cL_j$), and let $H$ be a uniform random hyperplane passing through the origin. Then the probability that $T$ is intersected by $\Nei(H,\delta)$ is $O(e^{-j\tau})$.
\end{restatable}

\begin{proof}
First, note that since $\cT$ is a nice tiling, the circumscribed ball diameter $\tau=2\rho\nu$ is a constant, and it is also an upper bound on the diameter of each tile.
Throughout this proof we work in the Poincar\'e ball model centered at the
origin $o$. Let $T\in \cL_j$ be a tile intersected by $\Nei(H,\delta)$. Let
$p\in T$ be a point of $T$ realizing the distance to the origin, i.e.,
$\dist(p,o)\in [(j-1)\tau,j\tau)$; see Figure~\ref{fig:intersectingtile}. Since
$\diam(T)\leq \tau$, there is a point $p'\in T\cap \Nei(H,\delta)$ at
distance at most $\tau$ from $p$. Since $p'\in \Nei(H,\delta)$, there is a point
$p''\in H$ where $\dist(p',p'')\leq \delta$.
It follows that
$\dist(p'',o)\in[(j-1)\tau-\delta,j\tau+\tau+\delta)$. Since $p''\in H$,
the point $q$ defined by the intersection of the sphere $\partial B(o,j\tau)$ and the line $op''$ satisfies
$\dist(q,p'')\leq \tau+\delta$.
Notice that $\dist(p',q)\leq \tau+2\delta$, and the
tile $T$ has diameter at most $\tau$, therefore
$T\subset B(q,\,2\tau+2\delta)$. Thus we have shown that if $T$ is intersected by $\Nei(H,\delta)$, then there is a point $q\in \partial B(o,j\tau)$ such that $B(q,\,2\tau+2\delta)$ is intersected by $\Nei(H,\delta)$, or equivalently, $H$ intersects the ball $B^\star \eqdef B(q,\,2\tau+3\delta)$. Clearly, it is sufficient to prove that the
probability that $H$ intersects $B^*$ is $O(e^{-j\tau})$.

\begin{figure}[t]
\centering
\includegraphics[scale=0.8]{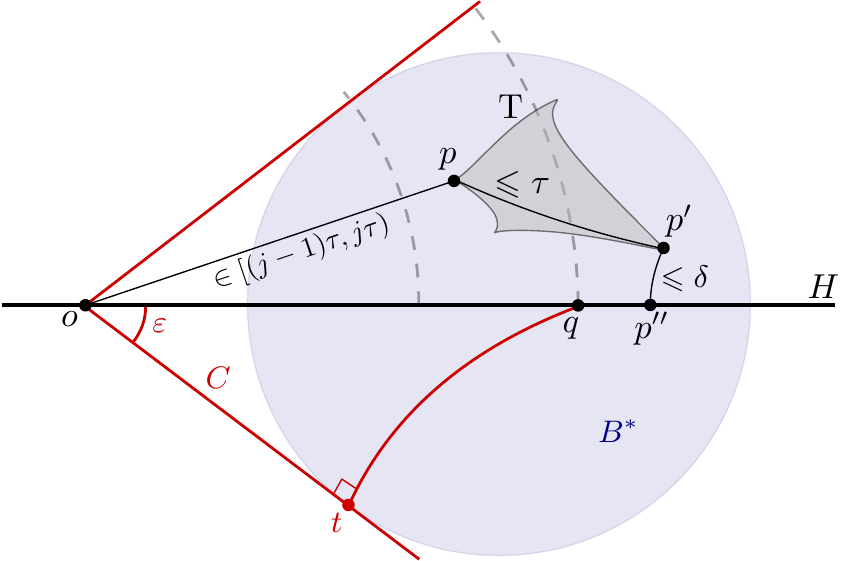}
\caption{A tile $T$ intersecting the neighborhood of the random line $H$ in the hyperbolic plane.}\label{fig:intersectingtile}
\end{figure}

Let $C$ be the cone (centered at $o$) touching $B^*$, and let $\eps$ be its half-angle. The hyperplane $H$ intersects the ball $B^*$ if and only if its normal is in the $\eps$-neighborhood of a great circle of the unit ($d-1$)-sphere, i.e., the normal of $H$ is in $\Nei_{S^{d-1}}(S^{d-2},\eps)$, where the distance on $S^{d-1}$ is the geodesic metric inherited from $\Reals^d$. In the $2$-flat given by $o,q$ and an arbitrary point $t$ where the line from $o$ touches $B^*$, the triangle $oqt$ has a right-angle at $t$, and $oq=j\tau$, and $qt=2\tau+3\delta$, therefore we have 

\[\eps =\arcsin\frac{\sinh(2\tau+3\delta)}{\sinh j\tau}<2\frac{\sinh(2\tau+3\delta)}{\sinh j\tau}=\frac{O(1)}{\Omega(e^{j\tau})} = O(e^{-j\tau}).\]

The $\eps$-neighborhood of a great circle in $S^{d-1}$ has volume
\begin{multline*}
V \eqdef \Vol(\Nei_{S^{d-1}}(S^{d-2},\eps)) = 2\int_{0}^{\sin \eps} \Vol\left(S^{d-2}\left(\sqrt{1-x^2}\right)\right)\mathrm{d}x\\
 < 2\int_{0}^{\sin \eps} \Vol(S^{d-2})\mathrm{d}x = O(\sin \eps) = O(\eps) = O(e^{-j\tau}),
\end{multline*}
where $S^{d-2}(\sqrt{1-x^2})$ is the $(d-2)$-dimensional Euclidean sphere of radius $\sqrt{1-x^2}$, or more formally, $S^{d-2}(\sqrt{1-x^2})\eqdef \bd B_{\Reals^{d-1}}(o,\sqrt{1-x^2})$.
Therefore the probability that $H$ intersects a given tile from $\cL_r$ is $V/\Vol(S^{d-1}) = O(e^{-j\tau})/\Omega(1)=  O(e^{-j\tau}).$
\end{proof}

\subsection{A hyperbolic separator theorem}

The following is the key lemma in the proof of Theorem~\ref{thm:sep}. It is inspired by the work of Fu~\cite{Fu2011}, who uses a similar approach to find clique-based separators in $\Reals^d$. For a tiling $\cT$ of $\Hyp^d$ and a given point set $P\subseteq \Hyp^d$, we say that a tile $T\in \cT$ is non-empty if $P\cap T \neq \emptyset$. The weight of a set $S$ of tiles (with respect to $P$) is defined as $\gamma(S)=\sum_{T\in S} \log(|P\cap T|+1)$. 

\begin{restatable}{lemma}{lemweight}\label{lem:weight}
Let $d\geq 2$ and $\delta>0$ be fixed constants, and let $P$ be a (multi)set of $n$ points in $\Hyp^d$, and let $p\in \Hyp^d$ be an arbitrary point. Let $\cT$ be a nice tiling of $\,\Hyp^d$, and let $H$ be a hyperplane through $p$ chosen uniformly at random. Let $S\subset \cT$ be the set of non-empty tiles intersected by the $\delta$-neighborhood of $H$, that is, $S\eqdef \{T\in \cT \mid T \cap P \neq \emptyset \text{ and } T\cap \Nei(H,\delta)\neq \emptyset\}$.
\begin{enumerate}
	\item[(i)] If $d=2$, then $\E[|S|]=O(\log n)$ and $\E[\gamma(S)] = O(\log^2 n)$.
	\item[(ii)] If $d\geq 3$, then $\E[|S|]=O(n^{1-1/(d-1)})$ and $\E[\gamma(S)]=O(n^{1-1/(d-1)})$.
\end{enumerate}
\end{restatable}

\begin{proof}
Let $\exp$ be the exponential function of base $e$. In this proof, we introduce some constants $c_1,c_2,\dots$ that may depend on $d$. If the parameters of the nice tiling $\cT$ are $\rho$ and $\nu$, then let $\tau\eqdef 2\rho\nu$ denote the circumscribed diameter. For a positive integer $j$, we set $\cL_j\eqdef\{T\in \cT \;|\; \dist(p,T)\in [(j-1)\tau, j\tau)\}$. Notice that the expected value of $\gamma(S)$ depends only on the distribution of points over the tiles. Let $n_T$ be the number of points from $P$ in the tile $T$. Note that for an empty tile $T$ we have $\log(n_T+1)=0$, thus
\[\E[\gamma(S)]=\sum_{T\in \cT} \log(n_T+1) \cdot \Pr[T \cap \Nei(H,\delta) \neq \emptyset]\\
 = \sum_{j=1}^{\infty} \; \sum_{T \in \cL_j} \log(n_T+1) \cdot \Pr[T \cap \Nei(H,\delta) \neq \emptyset].\]
By Lemma~\ref{lem:intersectingtile}, if $T\in \cL_j$, then $\Pr[T \cap \Nei(H,\delta) \neq \emptyset] \leq c_1e^{-j\tau}$ for some constant $c_1$, thus
\begin{equation}\label{eq:basicsum}
\E[\gamma(S)]\leq \sum_{j=1}^{\infty} \; \sum_{T \in \cL_r} c_1e^{-j\tau}\log(n_T+1).
\end{equation}

From now on, we concentrate on maximizing the function on the right hand side of~\eqref{eq:basicsum}, thus providing an upper bound on $\E[\gamma(S)]$.
Note that $c_1e^{-j\tau} \log(n_T+1)$ is a decreasing function of $j$. Suppose that $\dist(p,T)<\dist(p,T')$ and $n_T<n_{T'}$. The sum on the right hand side of~\eqref{eq:basicsum} is increased  or remains unchanged if we swap the content of $T$ and $T'$, i.e., move all points of $P\cap T$ to somewhere in $T'$ and vice versa. Without loss of generality, assume that $P$ is a point set where no such swaps are possible. Let $J$ be the largest index $j$ where $\cL_j$ contains a non-empty tile. Notice that if $j<J$, then all tiles in $\cL_{j}$ are non-empty. Since we have a ball of radius $(J-1)\tau$ that contains only non-empty tiles and at most $n$ points, this gives an upper bound on $J\tau$ by Proposition~\ref{prop:tilesballs} (i):

\begin{equation}\label{eq:Rbound}
J\tau \leq \frac{1}{d-1}(\ln n+c_2)
\end{equation}


Furthermore, if we fix $n_j=\sum_{T\in \cL_j} n_T$, the value of $\sum_{T\in \cL_j} \log(n_T+1)$ is maximized if the values $n_T$ are equal (i.e., we now allow fractional values for $n_T$). From now on, assume that for all $j$ and for all $T\in \cL_j$ we have $n_T=\frac{n_j}{|\cL_j|}$. By Proposition~\ref{prop:tilesballs} (ii), we know that $c_3e^{(d-1)j\tau} \leq |\cL_j| \leq c_4e^{(d-1)j\tau}$ for some constants $c_3,c_4$. Therefore if $T\in \cL_r$ with $r\leq R$, then we have $n_T\leq \frac{n_j}{c_3e^{(d-1)j\tau}}$. Hence,
\begin{align}
\begin{split}\label{eq:weightstop}
\E[\gamma(S)] &\leq \sum_{j=1}^J\sum_{T \in \cL_j} c_1e^{-j\tau}\log(n_T+1) \leq 
\sum_{j=1}^J |\cL_j| \cdot c_1e^{-j\tau}\log\left(\frac{n_j}{|\cL_j|}+1\right)\\
&\leq \sum_{j=1}^J c_1c_4e^{(d-2)j\tau}\,\log\left(\frac{n_j}{c_3e^{(d-1)j\tau}} + 1\right)
= O\left( \sum_{j=1}^J e^{(d-2)j\tau}\log\left(\frac{n_j}{e^{(d-1)j\tau}}+1\right)\right).
\end{split}
\end{align}

The expectation of $|S|$ can be bounded the same way, without the logarithmic term. For $d=2$, we get
\[\E[|S|]=O\left( \sum_{j=1}^J e^{(d-2)j\tau}\right)=O\left( \sum_{j=1}^J e^0\right)=O(J)=O(\log n)\]
by inequality~\eqref{eq:Rbound}. For the weight bound in case of $d=2$, we know that the weight of each tile is at most $\log (n+1)$, and $\E[|S|]=\log n$, therefore $\E[\gamma(S)]\leq \E[|S|]\log (n+1)=O(\log^2 n)$.
If $d\ge 3$, we have 

\begin{equation}\label{eq:sizebound3}
\E[|S|]\!=\!O\!\left(\sum_{j=1}^J e^{(d-2)j\tau}\!\right)\!
  \!=\! O\!\left( \exp\!\big((d-2)J\tau\big)\right)
  \!=\! O\!\left(\!\exp\!\left(\frac{d-2}{d-1}(\ln n + c_2)\right)\!\!\right)
 \!=\!O\!\left(n^{1-1/(d-1)}\right)\!.
 \end{equation}

We now bound the weight from~\eqref{eq:weightstop} in case of $d\ge 3$. Let $\eps>0$ be a small positive constant. We upper bound $n_j$ by $n$, and $\log(x+1)$ by $1+x^{\frac{d-2}{d-1}-\eps}$ (notice that the latter uses $d\ge 3$).

\begin{multline*}
\E[\gamma(S)] = O\left( \sum_{j=1}^J e^{(d-2)j\tau}\log\left(\frac{n_j}{e^{(d-1)j\tau}}+1\right)\right)
= O\left( \sum_{j=1}^J e^{(d-2)j\tau}\left(1+\left(\frac{n}{e^{(d-1)j\tau}}\right)^{\frac{d-2}{d-1}-\eps}\right)\right)\\
=n^{\frac{d-2}{d-1}-\eps} \cdot O\left(\sum_{j=1}^J \exp(\eps(d-1)j\tau)\right)+ O\!\left(\sum_{j=1}^J e^{(d-2)j\tau}\!\right)\!.
\end{multline*}
We apply the bound of~\eqref{eq:sizebound3} for the second term. By computing the sum in the first term and applying inequality~\eqref{eq:Rbound} we get

\begin{multline*}
\E[\gamma(S)] = n^{\frac{d-2}{d-1}-\eps} \cdot O\left(\exp\!\left(\eps(d-1)J\tau\right)\right)+ O(n^{1-1/(d-1)})\\
 =n^{\frac{d-2}{d-1}-\eps}  \cdot  O\left(\exp\left(\eps\left(\ln n +c_2\right)\right)\right)+ O(n^{1-1/(d-1)})
=O\left(n^{1-1/(d-1)}\right).%
\tag*{$\qed$} 
\end{multline*}
\renewcommand{\qedsymbol}{}
\end{proof}

The \emph{centerpoint} of a finite point set $P\subset \Hyp^d$ is a point $p\in \Hyp^d$ such that for any hyperplane $H$ through~$p$ the two open half-spaces with boundary $H$ both contain at most $\frac{d}{d+1}n$ points from $P$. Note that a centerpoint of $P$ is not necessarily a point in $P$.

\begin{lemma}\label{lem:balance}
Every finite point set of $\Hyp^d$ has a centerpoint, and it can be computed in polynomial time.
\end{lemma}

\begin{proof}
Let $V_{BK}$ be the point set corresponding to $V$ in the Beltrami-Klein model of $\Hyp^d$, that is, $V_{BK}$ is a point set in the unit ball around the origin in $\Reals^d$. Let $p\in D$ be the centerpoint of $V_{BK}$ in the traditional (Euclidean) sense. It follows that any hyperplane $H$ through $p$ is a $d/(d+1)$-balanced separator of $V_{BK}$, i.e., on both sides of $H$ there are at most $\frac{d}{d+1}n$ points from $V_{BK}$. In the Beltrami-Klein model, Euclidean hyperplanes are also hyperplanes in the hyperbolic sense, therefore $p$ is a centerpoint in the hyperbolic sense as well. Converting $V$ to $V_{BK}$ takes linear time. Computing the centerpoint in Euclidean space can be done in $O(n^d)$ time as a consequence of~\cite[Theorem 4.3]{edelsbrunner}. Converting the resulting point back to $\Hyp^d$ takes constant time.
\end{proof}

\begin{remark}
For certain applications, having the centerpoint computed more efficiently may be desirable. If a weaker balance factor for the separator theorem is sufficient, then computing an approximate centerpoint suffices. In such cases, one can use an $O(d^{O(\log d)}n)$ time deterministic algorithm for finding an approximate centerpoint~\cite{DBLP:journals/dcg/MulzerW13}.
\end{remark}

We are finally in a position to prove Theorem~\ref{thm:sep}.

\begin{proof}[Proof of Theorem~\ref{thm:sep}]
Let $G=(V,E)$ be our input graph with embedding $\eta$, and let $P\eqdef\eta(V)$.
We begin by computing a centerpoint of $P$ in $\poly(n)$ time according to Lemma~\ref{lem:balance}. Let $p$ be the resulting point.

Next, we fix a tiling $\cT$ with tile diameter $\rho-\eps$ for some small positive constant $\eps$, so that any pair of points in the same tile are necessarily connected; this tiling can be constructed using Lemma~\ref{lem:tilings} (ii). Let $\cP$ be the vertex partition of $G$ corresponding to this tiling, i.e., for each tile we create a partition class in $\cP$, which necessarily forms a clique in $G$. Let $H$ be a uniform random hyperplane through $p$, and let $\bar{H}$ denote the its $\rho\nu$-neighborhood: $\bar{H}\eqdef \Nei(H,\rho\nu)$.
Lemma~\ref{lem:weight} shows that the set $S$ of cliques given by tiles that intersect $\bar{H}$ has expected size and weight as desired. Due to the properties of a centerpoint, both sides of $\bar{H}$ contain at most $\frac{d}{d+1}n$ points. Moreover, the set $\bigcup S$ is a separator. To see this, assume for the purpose of contradiction that there is a pair $x,y\in P$ on different sides of $\bar{H}$ that are connected. Then the geodesic $xy$ has a unique intersection $q$ with $H$, and since $x,y \not \in \bar{H}$, both $qx$ and $qy$ are longer than $\rho\nu$. Hence, $\dist(x,y)>2\rho\nu$, so $x$ and $y$ cannot be connected; this is a contradiction.

By Markov's inequality, with probability $1/2$ a random separator will have at most twice the expected weight, and similarly with probability $1/2$ it has at most twice the expected size. Observe that $|S|\leq \gamma(|S|)\leq |S|\log(n+1)$, so in case of $d\ge 3$ having weight $O\big(n^{1-1/(d-1)}\big)$ guarantees that the size  is at most $O\big(n^{1-1/(d-1)}\big)$. For $d=2$, having size $O(\log n)$ guarantees that the weight is $O(\log^2 n)$. We can compute the size and weight of a separator in polynomial time. Consequently, we can find a separator of the desired size and weight in $\poly(n)$ expected time.
\end{proof}


\section{Treewidth bounds and algorithmic applications}\label{sec:twhyp3}

\paragraph*{Partitions and treewidth basics.}\label{subsec:twbasic}
Let us begin with some basic notions and terminology related to $\cP$-flattened treewidth.
Let $G=(V,E)$ be a simple graph. A \emph{clique-partition} of $G$ is a partition $\cP$ of $V$ where each partition class $C\in \cP$  forms a clique in $G$. A \emph{$\kappa$-partition} is a partition $\cP$ of $V$ where each partition class $C\in \cP$ induces a connected subgraph of $G$ that can be covered by $\kappa=O(1)$ cliques; for example, clique partitions are $1$-partitions. It has been observed in~\cite{frameworkpaper} that $\kappa$-partitions with $1<\kappa=O(1)$ can be used instead of clique partitions for many algorithms; we will also make use of this option.

For a graph $G$, let $I\subset V(G)$ be a maximal independent set. We create a partition class for each $v\in I$, and assign vertices in $V(G)\setminus I$ to the class of an arbitrary neighbor inside $I$. A partition $\cP$ created this way is called a \emph{greedy partition}. Note that greedy partitions can be found in polynomial time in any graph. For example, greedy partitions are $\kappa$-partitions in constant-dimensional Euclidean unit ball graphs~\cite{frameworkpaper}.

The \emph{$\cP$-contraction} of $G$ is the graph obtained by contracting all edges induced by each partition class, and removing parallel edges; it is denoted by $G_\cP$. 
The \emph{weight} of a partition class $C$ is defined as $\log(|C|+1)$. Given a set $S\subset \cP$, its weight $\gamma(S)$ is defined as the sum of the class weights within, i.e., $\gamma(S)\eqdef \sum_{C\in S} \log(|C|+1)$. Note that the weights of the partition classes define vertex weights in the contracted graph $G_\cP$.

A \emph{tree decomposition} of a graph $G=(V,E)$
is a pair $(T,\sig)$ where $T$ is a tree and $\sig$ is a mapping from the vertices of $T$ to subsets of $V$
called \emph{bags}, with the following properties. Let $\bags(T,\sig) \eqdef \{
\sig(u): u \in V(T) \}$ be the set of bags associated to the vertices of $T$. Then
we have: (1) For any vertex $u\in V$ there is at least one bag in $\bags(T,\sig)$ containing it. (2) For any edge $(u,v)\in E$ there is at least one bag in $\bags(T,\sig)$ containing both $u$ and $v$. (3) For any vertex $u\in V$ the collection of bags in $\bags(T,\sig)$ containing $u$ forms a subtree of~$T$.
The \emph{width} of a tree decomposition is the size of its
largest bag minus~1, and the \emph{treewidth} of a graph $G$ equals the minimum width of a tree
decomposition of $G$. 
We talk about pathwidth and path decomposition if $T$ is a path.

We will need the notion of {\em weighted treewidth}~\cite{EijkhofBK07}.
Here each vertex has a weight, and the
{\em weighted width} of a tree decomposition is the maximum over the bags of the
sum of the weights of the vertices in the bag (note: without the $-1$).
The {\em weighted treewidth} of a graph is the minimum weighted width over its tree decompositions.
Now let $\cP$ be a partition of a given graph $G$.
We apply the concept of weighted treewidth
to $G_{\cP}$, where we assign each vertex $C$ of $G_\cP$ the weight~$\log(|C|+1)$, and refer to this weighting whenever we talk about the weighted treewidth of a contraction $G_\cP$. For any given $\cP$, the weighted treewidth of $G_\cP$ with the above weighting is referred to as the \emph{$\cP$-flattened treewidth} of $G$.

\subsection{Treewidth in \texorpdfstring{$\Hyp^2$}{the hyperbolic plane}}

Although our bound on the weight of the separator in $\Hyp^2$ in Theorem~\ref{thm:sep} is optimal up to constant factors ---it is attained for the input where each tile in a hyperbolic disk of radius $\sim\ln(\sqrt{n})$ contains $\sqrt{n}$ points--- the bound for $\Hyp^2$ in Corollary~\ref{cor:shallowpathwidth} is far from optimal. We could use Theorem~\ref{thm:sep} to directly design a divide-and-conquer algorithm for \IS, and the running time would be $2^{O(\log^3 n)}$ in $\NUBG_{\Hyp^2}(\rho,\nu)$, or $2^{O(\log^2 n)}$ in $\SNUBG_{\Hyp^2}(\rho,\nu)$. Both of these running times can be significantly improved by a better bound on treewidth, as stated in Theorem~\ref{thm:treewidth}.
We prove Theorem~\ref{thm:treewidth} in three stages: first, we show that the neighborhood graph of a finite tile set has treewidth $O(\log n)$, then we transfer this bound to shallow graphs, and finally, using the following lemma, we transfer the bound to non-shallow graphs.


\begin{restatable}{lemma}{lemGPisSDTG}\label{lem:G_PisSDTG}
Consider a graph $G\in \NUBG_{\Hyp^d}(\rho,\nu)$ with some fixed embedding, and let $\cP$ be the partition of $V(G)$ given by a $(\rho_\cT,\nu_\cT)$-nice tiling of ${\Hyp^d}$. Then $G_\cP\in \SNUBG_{\Hyp^d}(\rho,\nu',k)$ where $\nu'=\nu+2\rho_\cT\nu_\cT/\rho$ and $k=\Vol(\rho\nu'+2\rho_\cT\nu_\cT)/\Vol(\rho_\cT)$, where $\Vol(r)$ denotes the volume of a ball of radius $r$ in $\Hyp^d$. 
\end{restatable}

\begin{proof}
First, we define an embedding $\eta_\cP$ for $G_\cP$ by assigning each vertex $C\in \cP$ to a point $p_C \in \eta(C)$. We show that $G_\cP\in \SNUBG_{\Hyp^d}(\rho,\nu',k)$ for the embedding $\eta_\cP$.

Clearly any pair of vertices at distance less than $2\rho$ is connected. Now suppose for contradiction that a pair of points $p_i$ and $p_j$ in tiles $T_i$ and $T_j$ have distance more than $2\rho\nu'=2\rho\nu+4\rho_\cT\nu_\cT$. Observe that $T_i \subset B(p_i,2\rho_\cT\nu_\cT)$, and similarly $T_j\subset B(p_j,2\rho_\cT\nu_\cT)$. Since $\dist(p_i,p_j)>2\rho\nu+4\rho_\cT\nu_\cT$, any pair of points in $B(p_i,2\rho_\cT\nu_\cT)$ and $B(p_j,2\rho_\cT\nu_\cT)$ must have distance at least $2\rho\nu$. Hence, no pair of points in $T_i$ and $T_j$ is connected, so the corresponding vertices in $G$ are also not connected. But this contradicts that $G\in \NUBG_{\Hyp^d}(\rho,\nu)$ for the embedding $\eta$.

To prove that  $\eta_\cP$ is a shallow embedding, consider a point $x\in {\Hyp^d}$ with a circumscribed ball $B(x,\rho\nu')$. The points of $\eta(\cP)$ intersected by this ball all lie in different tiles of $\cT$. The inscribed balls of all of these tiles are centered somewhere in $B(x,\rho\nu',2\rho_\cT\nu_\cT)$, and these balls are disjoint, therefore there cannot be more than
  $\Vol(\rho\nu'+2\rho_\cT\nu_\cT)/\Vol(\rho_T)=k$ such inscribed balls. It follows that $|B(x,\rho\nu')\cap \eta(\cP)|\leq k$.
\end{proof}


The first stage of the proof (for neighborhood graphs) relies crucially on the isoperimetric inequality~\cite{teufel1991generalization}, which states that for a simple closed curve of length $L$ bounding an area $A$ in the hyperbolic plane $\Hyp^2$, we have $L^2 \geq 4\pi A + A^2$
where equality is attained only for a geodesic circle. In fact, we only need the following simple corollary.

\begin{corollary}\label{cor:isoperim}
Let $\phi$ be a simple closed curve in $\Hyp^2$, and let $\phi_\myin$ denote the region enclosed by $\phi$. Then $\len(\phi)\geq \mathrm{area}(\phi_\myin)$.
\end{corollary}

Consider a tiling $\cT$ given by Lemma~\ref{lem:tilings}(i). Notice that the neighborhood graph of any tile set $S\subset \cT$ is planar. A \emph{hole} for a set of tiles $S\subset \cT$ is a finite set $S'\subset \cT\setminus S$ such that $\partial S'$ is a closed curve and $\partial S' \subset \partial S$.
An \emph{outerplanar} or $1$-\emph{outerplanar} graph is a planar graph that has a plane embedding where all vertices lie on the outer face. For $k\geq 2$, a $k$-\emph{outerplanar} graph is a plane graph (i.e., a planar graph with a fixed embedding) where removing vertices of the outer face leads to a $(k-1)$-outerplanar graph.

\begin{lemma}\label{lem:layerpeel}
Let $\cT$ be a regular tiling of $\,\Hyp^2$ as constructed in Lemma~\ref{lem:tilings}. Then the neighborhood graph of any finite tile set $S\subset \cT$ is
$(c\log |S|)$-outerplanar, where $c$ is an absolute constant independent of our choice of $\cT$.
\end{lemma}

\begin{proof}
Let $\alpha$ and $\xi$ be the area and circumference of a tile in $\cT$ respectively, and let $G$ be the neighborhood graph of $S$. Let $\bar{S}$ be the tile set obtained from $S$ by filling all of its holes with tiles. We will prove by induction on $|S|$ that the outerplanarity of $G$ is at most $\left(\frac{\log |S|}{\log(1/(1-\alpha/\xi))}\right)$. To prove this claim, we show that a constant proportion of all tiles are adjacent to the unbounded component of $\cT\setminus S$.

The boundary of $\bigcup \bar{S}$ is a collection of closed curves bounding interior disjoint regions whose union has precisely $|\bar{S}|$ tiles. So by summing the inequality $\len(\phi)\geq \mathrm{area}(\phi_\myin)$ for each of these curves, we have that $\len(\partial(\bigcup\bar{S})) \geq \alpha|\bar{S}|$.
Let $S'$ be the set of tiles in $S$ adjacent to the unbounded component of $\cT\setminus S$. The total length of $\partial(\bigcup\bar{S})$ is at most $|S'|\xi$, so $|S'|\xi \geq \alpha|\bar{S}|$. Hence $|S'| \geq  \frac{\alpha}{\xi}|\bar{S}| \geq \frac{\alpha}{\xi} |S|$, as required.

The neighborhood graph of $S\setminus S'$ has outerplanarity at most $\frac{\log(|S|-|S'|)}{\log(1/(1-\alpha/\xi))}$ by the induction hypothesis. Therefore the outerplanarity of $G$ is at most
 \begin{equation}\label{eq:outp}
 \frac{\log(|S|-|S'|)}{\log(1/(1-\alpha/\xi))}+1
 \leq \frac{\log(|S|-\frac{\alpha}{\xi}|S|)}{\log(1/(1-\alpha/\xi))}+1
 = \frac{\log |S|}{\log(1/(1-\alpha/\xi))}.
 \end{equation}
 Recall that in Lemma~\ref{lem:tilings}(i) we defined a tiling $\cT_\delta$ for any choice of $\delta\in \Nats^+$; let us denote the corresponding tile area and diameter by $\alpha_\delta$ and $\xi_\delta$ respectively.
 By  Corollary~\ref{cor:isoperim}, we have $0<\alpha_\delta/\xi_\delta<1$, but this fact by itself is not enough to establish the required bound. Indeed, in principle we could have
 $\inf_{\delta\in \Nats^+} \alpha_\delta/\xi_\delta=0$,
 which would make the outerplanarity bound established in~\eqref{eq:outp} depend on the choice of $\cT$.
 But we have shown in Lemma~\ref{lem:tilings}(i) that $\lim_{\delta\rightarrow \infty} \alpha_\delta/\xi_\delta>0$, so there exists a constant $c'>0$ such that $\alpha_\delta/\xi_\delta>c'$ for all $\delta$. This concludes our proof.
\end{proof}

The treewidth of $k$-outerplanar graphs is at most $3k-1$ by~\cite[Theorem 83]{Bodlaender98}. By Lemma~\ref{lem:layerpeel}, this implies that the neighborhood graph of any tile set $S\subset \cT$ has treewidth $O(\log n)$. We remark that this is the only result in the present paper that also works for tilings of non-constant diameter. As a consequence, it can be shown that any $n$-vertex subgraph of a regular tiling with bounded tiles has treewidth at most $c\log n$, where $c$ is independent of the tiling, that is, this holds even if the choice of the regular tiling depends on $n$.

In order to transfer the treewidth bound from neighborhood graphs to graphs in $\SNUBG_{\Hyp^2}(\rho,\nu,k)$, we need two operations: the $k$-fold blowup and $k$-th power. It turns out that both operations increase the treewidth by at most a function of $k$ and the maximum degree of $G$.
The \emph{$k$-th power} $G^k$ of a graph $G$ is the graph with vertex set $V(G)$ where $u,v\in V(G)$ are connected if and only if $\dist_G(u,v)\leq k$.

\begin{restatable}{lemma}{lempowertw}\label{lem:power_tw}
If $G$ has maximum degree $\Delta$, then $\tw(G^k)\leq \Delta^{\ceil{k/2}+1}(\tw(G)+1)-1$.
\end{restatable}

\begin{proof}
Let $(T,\sig)$ be a tree decomposition of $G$, and replace each bag $B$ with
$\Nei_G(B,\ceil{k/2})$. It is sufficient to prove that the resulting bags form a tree
decomposition of $G^k$, since the new bags are at most $\Delta^{\ceil{k/2}+1}$ times
larger than the bags of $(T,\sig)$. The first property (every vertex is
contained in some bag) trivially holds. The second property (every edge is induced by some bag) is true since the endpoints of any edge $uv$ in $G^k$ can be connected by a $u-v$ path of length at most $k$ in $G$, so if $B$ is a bag in the tree decomposition of $G$ that contains a midpoint $w$ of this path, then $\Nei_G(B,\ceil{k/2})$ also contains $u$ and $v$.

Next, we prove the third property, which is equivalent to the following: if
vertex $v$ appears in the bags $B$ and $B'$, then it appears in every bag that
is on the unique tree path between $B$ and $B'$. Suppose for a contradiction
that there is a vertex $v\in V(G)$ that appears in bags $\Nei_G(B,k)$ and
$\Nei_G(B',\ceil{k/2})$, while there is a bag $\Nei_G(B^*,\ceil{k/2})$ between them in the tree
not containing $v$. Since $v\in \Nei_G(B,\ceil{k/2})$, it follows that $\Nei_G(v,\ceil{k/2})\cap
B \neq \emptyset$ and similarly $\Nei_G(v,\ceil{k/2})\cap B' \neq \emptyset$. Also,
since $v\not \in \Nei_G(B^*,\ceil{k/2})$, we have that $\Nei_G(v,\ceil{k/2})\cap B^*=\emptyset$.
Since $(T,\sig)$ is a valid tree decomposition, $B^*$ is a separator of $G$
where $B$ and $B'$ are in distinct connected components of $G[V(G)\setminus
B^*]$. But this contradicts the fact that $\Nei_G(v,\ceil{k/2})$ is a connected graph
disjoint from $B^*$ that contains vertices from both $B$ and $B'$.
\end{proof}

For $k\in \Nats_+$ the \emph{$k$-fold blowup} of a graph $G$ is a graph $G^{(k)}$ where each vertex $v$ of $G$ is replaced by a $k$-clique $C_v$, and for each edge $uv\in E(G)$ and for all pairs of vertices $u'\in C_u, v'\in C_v$ we have $u'v'\in E(G^{(k)})$. Since the blown-up bags (replacing each vertex in each bag with the corresponding clique) give a tree decomposition of the blown-up graph, the blowup has treewidth at most $k(\tw(G)+1)-1$. We can summarize this in the following proposition.

\begin{proposition}\label{prop:blowup_tw}
For any graph $G$ its $k$-fold blowup satisfies $\tw(G^{(k)}) \leq k\tw(G) + k -1$.
\end{proposition}

We can now prove Theorem~\ref{thm:treewidth}.

\begin{proof}[Proof of Theorem~\ref{thm:treewidth}]
(i) Given an embedding of a shallow graph $G$, the main idea of this proof is to take a tiling, and to consider some power of the neighborhood graph of the non-empty tiles. This graph power can be blown up in a way that the resulting graph $H$ is a supergraph of our shallow graph $G$. Since neighborhood graphs have logarithmic treewidth, the resulting graph $H$ will also have treewidth $O(\log n)$, and since $H$ is a supergraph of $G$, the treewidth of $H$ is greater or equal to the treewidth of $G$. Next we give the details of the proof.

Let $P$ be the point (multi)set of the embedding of $G$, and let us refer to points of $P$ being connected or not connected if the corresponding vertices in $G$ are connected by an edge (resp. not connected). We fix a tiling $\cT$ from those offered by Lemma~\ref{lem:tilings}(ii) where the tile diameter is the smallest  possible tile diameter greater than $2\rho\nu$. Let $\delta$ denote the diameter of the tiles in $\cT$. Note that each tile contains at most $k'\eqdef k\Vol(\delta+\rho)/\Vol(\rho)=O(1)$ points of $P$. Let $S$ be the set of non-empty tiles in $\cT$. By Corollary~\ref{cor:nb_tiles}, there is a constant $c$ such that the tiles of $\Nei_{\cT}(S,1)$ cover $\Nei_{\Hyp^2}(\bigcup S,c)$. Applying the same corollary to $\Nei_{\cT}(S,1)$, we get that the tiles of $\Nei_{\cT}(S,2)=\Nei_{\cT}(\Nei_{\cT}(S,1),1)$ cover $\Nei_{\Hyp^2}(\bigcup S,2c)$. By induction we get that the tiles of $\Nei_\cT(S,r)$ cover $\Nei_{\Hyp^2}(\bigcup S,rc)$, Consequently, any pair of tiles $T,T' \in S$ for which $\dist_{\cT}(T,T')>2\rho\nu/c$ satisfies\footnote{We let $\dist_{\Hyp^2}(T,T')\eqdef\inf\{\dist_{\Hyp^2}(x,y) \;|\; x\in T, y\in T')\}$. Note that this is not a metric on subsets of $\Hyp^2$.} $\dist_{\Hyp^2}(T,T')>2\rho\nu$. Therefore $T$ and $T'$ cannot contain a pair of connected points $p\in T\cap P, p'\in T'\cap P$. Let $\ell\eqdef\ceil{2\rho\nu/c}+1$. Then $S'\eqdef \Nei_{\cT}(S,\ell)$ has the property that for any pair of connected points in $P$, their containing tiles $T$ and $T'$ are not too distant in the neighborhood graph of $S'$: $\dist_{S'}(T,T') <\ell$. Note that $|S'|=O(n)$. By Lemma~\ref{lem:layerpeel}, the neighborhood graph $\NbG(S')$ is $O(\log|S'|)=O(\log n)$-outerplanar, and by~\cite{Bodlaender98}, it has treewidth $O(\log n)$.

Consider the graph $(\NbG(S'))^\ell$. A pair of points in $P$ can be connected only if they are in the same tile or their tiles are connected in $(\NbG(S'))^\ell$. Let $H$ be the $k'$-fold blowup of $(\NbG(S'))^\ell$. Since each tile has at most $k'$ points of $P$, the graph $H$ is a supergraph of $G$. By Lemma~\ref{lem:power_tw}, $(\NbG(S'))^\ell$ has treewidth at most $\Delta^{\ceil{\ell/2}+1}(\tw(\NbG(S'))+1)-1$ where $\Delta$ is the number of neighboring tiles to a tile in $\cT$; therefore, $\tw((\NbG(S'))^\ell)=O(\tw(\NbG(S')))=O(\log n)$, and its blowup $H$ also has treewidth $O(\log n)$ by Proposition~\ref{prop:blowup_tw}. Since $H$ is a supergraph of $G$, $G$ has treewidth $O(\log n)$.

(ii) The weighted treewidth of $G_\cP$ is at most $\log(n+1)\cdot(\tw(G_\cP)+1)$ since each node in $G_{\cP}$ has weight at most $\log(n+1)$, so Lemmas~\ref{lem:tilings},~\ref{lem:G_PisSDTG} and Theorem~\ref{thm:treewidth}(i) imply the desired bound.
\end{proof}

\section{Computing treewidth and \texorpdfstring{$\cP$}{P}-flattened treewidth}\label{sec:gettw}
We can remove the dependence on the embedding from Theorem~\ref{thm:treewidth} by the same techniques as in~\cite{frameworkpaper}. The algorithms obtained this way are also deterministic since they do not rely on the Las Vegas separator-finding algorithm of Theorem~\ref{thm:sep}.


\paragraph*{Treewidth bounds for shallow graphs in \texorpdfstring{$\Hyp^d$ ($d\geq 3$) and in $\Hyp^2$}{in hyperbolic space}.}
The size bound of our separator theorem can be used to get a bound on the pathwidth (and treewidth) of shallow graphs, since in a shallow graph each tile has $O(1)$ points, and the weight of each non-empty tile is $\Theta(1)$. Thus the bounds that we proved in the previous section for clique-based separators immediately give the same asymptotic bounds for normal separators in shallow graphs. To turn these bounds into bounds on treewidth, we only need Theorem 20 of \cite{Bodlaender98}, which yields the following corollary:

\begin{corollary}\label{cor:shallowpathwidth}
Fix the constants $\rho$, $\nu\geq 1$ and $k$. Then for any $G\in \SNUBG_{\Hyp^2}(\rho,\nu,k)$ on $n$ vertices the pathwidth is $O(\log^2 n)$ and if $d\ge 3$, then for any $G\in \SNUBG_{\Hyp^d}(\rho,\nu,k)$ the pathwidth is $O(n^{1-1/(d-1)})$.
\end{corollary}

Following the techniques of~\cite{frameworkpaper} one can derive the bound $O(n^{1-1/(d-1)})$ for the $\cP$-flattened treewidth of graphs in $\NUBG_{\Hyp^d}(\rho,\nu)\; (d\geq 3)$ as well. Before we do that, we establish a way to get useful partitions in the absence of an embedding.

\begin{lemma}\label{lem:greedypart}
Let $\rho>0$, $\nu\geq 1$ and $d\ge 2$ be fixed constants, and let $G\in \NUBG_{\Hyp^d}(\rho,\nu)$. Then all greedy partitions $\cP$ of $G$ are $\kappa$-partitions for some $\kappa=O(1)$, and $G_\cP$ has maximum degree $O(1)$. Moreover, $G_\cP$ is shallow: $G_\cP\in \SNUBG_{\Hyp^d}(\rho,3\nu,2)$
\end{lemma}

\begin{proof}
Let $\eta$ be an embedding of $G$, and let $\cP$ be a greedy partition built around a maximal independent set $I\subset V(G)$. First, we claim that restricting $\eta$ to $I$ gives a representation for $G_\cP$ as a graph in $\SNUBG_{\Hyp^d}(\rho,3\nu,2)$. Pairs of points from $\eta(I)$ closer than $\rho$ are always connected in $G_\cP$, since there are no such point pairs. Pairs of points $p,q\in \eta(I)$ more than $6\rho\nu$ distance away cannot have neighbors $p'$ and $q'$ that are connected, so the corresponding partition classes are not connected in $G_\cP$. Finally, any pair of points from $\eta(I)$ have distance at least $2\rho$, so a ball of radius $\rho$ can contain at most two points of $\eta(I)$. This shows that $G_\cP\in \SNUBG_{\Hyp^d}(\rho,3\nu,2)$.

Clearly each partition class induces a connected subgraph of $G$. To bound $\kappa$, notice that the points of a partition class for point $v\in I$ all lie within $B(\eta(v),2\rho\nu)$. Consider a tiling $\cT$ from Lemma~\ref{lem:tilings}(ii) of diameter $\rho-\epsilon$ for some small $\epsilon>0$. Each tile induces a clique in $G$. The tiles intersecting $B(\eta(v),2\rho\nu)$ are contained in $B(\eta(v),2\rho\nu +\rho)$, and this ball contains $\Theta(e^{(d-1)(2\rho\nu +\rho)})=O(1)$ tiles by Proposition~\ref{prop:tilesballs} (i). Therefore each partition class can be covered by at most $\kappa=O(1)$ cliques.

To get a bound on the maximum degree in $G_\cP$, recall that the balls of radius $\rho$ with centers in $\eta(I)$ are disjoint. Since $G_\cP\in \SNUBG_{\Hyp^d}(\rho,3\nu,2)$, the neighborhood of a vertex is within a ball of radius $6\rho\nu$ around the representing point. Therefore the maximum degree is at most $\Vol(6\rho\nu)/\Vol(\rho))=O(1)$, where $\Vol(x)$ denotes the volume of a ball of radius $x$.
\end{proof}

We can now make our separator theorem work for greedy partitions.

\begin{lemma}\label{lem:sep_norep}
Let $\rho>0$, $\nu\geq 1$ and $d\ge 2$ be fixed constants, and let $G\in \NUBG_{\Hyp^d}(\rho,\nu)$ with representation $\eta$. Then for any greedy partition $\cP$ there is a $d/(d+1)$-balanced separator $S\subset \cP$ such that $|S|=O(\log n)$ and has weight $\gamma(S)=O(\log^2 n)$ if $d=2$, and $|S|=O(n^{1-1/(d-1)})$ and has weight $\gamma(S)=O(n^{1-1/(d-1)})$ if $d\ge 3$.
\end{lemma}

\begin{proof}
Let $\cP$ be a greedy partition for a given maximal independent set $I$. Let $\bar{G}_\cP$ be the graph obtained from $G$ by adding an edge between any pair of vertices $v,w$ that are in the same partition class. Thus $\cP$ is a clique-partition of $\bar{G}$, and $G$ is a subgraph of $\bar{G}$. Hence, it is sufficient to provide a separator of the desired asymptotic weight and size in $\bar{G}$. 

By Lemma~\ref{lem:greedypart}, $G_\cP \in \SNUBG_{\Hyp^d}(\rho,3\rho\nu,2)$. Let $\eta_\cP$ be the corresponding embedding of $G_\cP$, and let $P_\cP$ denote its image. Let $\bar{P}$ be the point multiset that for each vertex $C\in \cP$ contains $|C|$ copies of the point $\eta_\cP(C)$. Let $\bar{\eta}:V(G)\rightarrow \bar{P}, \; \bar{\eta}(v)=\eta_\cP(C_v)$, where $C_v\in \cP$ is the class of $\cP$ that contains $v$. Clearly $\bar{\eta}$ is a representation of $\bar{G}$ as a graph in $\NUBG_{\Hyp^d}(\rho,3\rho\nu)$.

If we take a nice tiling of $\Hyp^d$ as in Lemma~\ref{lem:weight}, then the resulting partition $\cP'$ will be a coarsening of $\cP$, and a separator $S$ consisting of $\cP'$-classes naturally decomposes into a separator consisting of $\cP$-classes. Moreover, each $\cP'$-class can contain at most $O(1)$ classes of $\cP$, since $P_\cP$ is sparse. Therefore a separator $S'$ for $\cP'$ of weight $\gamma$ decomposes into a separator $S$ of weight at most $O(\gamma)$. Thus (by Lemma~\ref{lem:weight}) $\bar{G}$ has a separator $S\subseteq \cP$ of the desired weight, which is also a separator of $G$.
\end{proof}

We are now ready to prove Theorem~\ref{thm:main}.

\begin{restatable}{theorem}{thmmain}\label{thm:main}
Let $d\ge 2$, $k>0$, $\rho>0$ and $\nu\geq 1$ be constants.
\begin{enumerate}
\item[(i)] If $G\in \SNUBG_{\Hyp^2}(\rho,\nu,k)$, then the
treewidth of $\,G$ is $O(\log n)$, and a  tree decomposition of width $O(\log n)$ can be computed in polynomial time.
\item[(ii)] If $G\in \SNUBG_{\Hyp^d}(\rho,\nu,k)$, then the
treewidth of $\,G$ is $O(n^{1-1/(d-1)})$, and a  tree decomposition of width $O(n^{1-1/(d-1)})$ can be computed in $2^{O(n^{1-1/(d-1)})}$ time.
\item[(iii)] If $G\in \NUBG_{\Hyp^2}(\rho,\nu)$, then for any greedy partition $\cP$ the $\cP$-flattened treewidth of $G$ is
$O(\log^2 n)$, and a  weighted tree decomposition of width $O(\log^2 n)$ can be computed in $n^{O(\log n)}$ time.
\item[(iv)] If $G\in \NUBG_{\Hyp^d}(\rho,\nu)$, then for any greedy partition $\cP$ the $\cP$-flattened treewidth of $\,G$ is
$O(n^{1-1/(d-1)})$, and a  weighted tree decomposition of width $O(n^{1-1/(d-1)})$ can be computed in $2^{O(n^{1-1/(d-1)})}$ time.
\end{enumerate}%
\end{restatable}

\begin{proof}
Given a graph $G$ of treewidth $w$, we can use the algorithm from either \cite{Robertson95} or \cite{Bodlaender16} to compute a tree decomposition whose width is $O(w)$ in $2^{O(w)}\poly(n)$ time. Moreover, for a partition $\cP$ where the $\cP$-flattened treewidth is $w'$ we can also compute a weighted tree decomposition of width $O(w')$ in $2^{O(w')}\poly(n)$ time; see Lemma 8 and the proof of Theorem 9 in~\cite{frameworksArxiv}.
Putting Lemmas~\ref{lem:greedypart},~\ref{lem:G_PisSDTG} and Theorem~\ref{thm:treewidth} together concludes the proof of Theorem~\ref{thm:main} in case of $d=2$.

For $\Hyp^d,\, d\ge 3$ we can combine Lemma~\ref{lem:sep_norep} with arguments in the proof of Lemma 7 in~\cite{frameworksArxiv} to show that the $\cP$-flattened treewidth of $G$ is $O\big(n^{1-1/(d-1)}\big)$. The weighted treewidth approximation technique can produce a weighted tree decomposition of $G_\cP$ of width $O\big(n^{1-1/(d-1)}\big)$ in $2^{O(n^{1-1/(d-1)})}$ time, which concludes the proof for the case $d\geq 3$.
\end{proof}

\section{Algorithmic applications}\label{sec:algs}

In this section we showcase some algorithms that can be obtained from Theorem~\ref{thm:main}.

\begin{theorem}
Let $d\ge 2$, $k>0$, $\rho>0$ and $\nu\geq 1$ be constants. Then \IS can be solved in
\begin{itemize}
	\item $\poly(n)$ time in $\SNUBG_{\Hyp^2}(\rho,\nu,k)$
	\item $n^{O(\log n)}$ time in $\NUBG_{\Hyp^2}(\rho,\nu)$
	\item $2^{O(n^{1-1/(d-1)})}$ time in $\NUBG_{\Hyp^d}(\rho,\nu)$ if $d\ge 3$
\end{itemize}
\end{theorem}

\begin{proof}[Proof sketch]
If $G\in \SNUBG_{\Hyp^2}(\rho,\nu)$, then we can compute a tree decomposition of width $O(\log n)$ in polynomial time by Theorem~\ref{thm:main}, and apply a traditional algorithm of running time $2^{\tw}\poly(n)$, see~\cite{fptbook} for an example. This yields a polynomial algorithm.

Now suppose $G\in \NUBG_{\Hyp^d}(\rho,\nu)$. We compute a greedy partition $\cP$ of $G$ in polynomial time. By Theorem~\ref{thm:main}, we can find a weighted tree decomposition $(T,\sig)$ of $G_\cP$ of width $w_2=O(\log n)$ if $d=2$ or of width $w_d=O(n^{1-1/(d-1)})$ if $d\ge 3$ in $2^{O(w_d)}\poly(n)$ time. From this point we proceed just as in Section 2.3 and Theorem 2.8 in~\cite{frameworkpaper}: the weighted tree decomposition of $G_\cP$ can be transferred into a so-called \emph{traditional tree decomposition} of $G$, on which a treewidth-based algorithm \cite{fptbook} can be run with a small modification, namely that only those partial solutions need to be considered that select at most $\kappa$ vertices from each partition class.
\end{proof}

\begin{theorem}\label{thm:qcoloringalg}
Let $\rho>0$, $\nu\geq 1$, $2\leq d\in \Nats$, and $q\in \Nats$ be fixed constants. Then the \textsc{$q$-Coloring} problem can be solved in
\begin{itemize}
	\item $\poly(n)$ time in $\NUBG_{\Hyp^2}(\rho,\nu)$
	\item $2^{O(n^{1-1/(d-1)})}$ time in $\NUBG_{\Hyp^d}(\rho,\nu)$ if $d\ge 3$.
\end{itemize}
\end{theorem}

\begin{proof}
Let $\cP$ be a greedy partition.
If some clique of $G$ contains at least $q+1$ vertices, then there is no $q$-coloring. Since $\cP$ is a $\kappa$-partition for some constant $\kappa$ by Lemma~\ref{lem:greedypart}, if a partition class has more than $\kappa q$ vertices, then we can reject. Otherwise $G\in \SNUBG_{\Hyp^d}(\rho,\nu,\kappa q+1)$. If $d=2$, then by Theorem~\ref{thm:main} we can find a tree decomposition of $G$ of width $w_2=O(\log n)$ in polynomial time. If $d\ge 3$, then the tree decomposition of $G_\cP$ found by Theorem~\ref{thm:main} has width $w_d=O(n^{1-1/(d-1)})$ and it can be found in $2^{O(n^{1-1/(d-1)})}$ time. Then we can simply run a treewidth-based coloring algorithm~\cite{fptbook} with running time $q^{\tw}\tw^{O(1)}\poly(n)$, which for $q=O(1)$ becomes $\poly(n)$ time for $d=2$ and $2^{O(n^{1-1/(d-1)})}$ time for $d\geq 3$.
\end{proof}

\begin{theorem}\label{thm:hcalg}
Let $\rho>0$, $\nu\geq 1$, and $2\leq d\in \Nats$ be fixed constants. Then \textsc{Hamiltonian Cycle} can be solved in
\begin{itemize}
	\item $\poly(n)$ time in $\NUBG_{\Hyp^2}(\rho,\nu)$ if an \emph{embedding is given}
	\item $2^{O(n^{1-1/(d-1)})}$ time in $\NUBG_{\Hyp^d}(\rho,\nu)$ if $d\ge 3$ and the \emph{embedding is given}.
\end{itemize}
\end{theorem}

\begin{proof}
Let $\cT$ be the tiling of Lemma~\ref{lem:tilings}(ii) with diameter $2\rho-\epsilon$, where $\eps$ is small positive constant. Notice that vertices of $G$ assigned to the same tile are connected in $G$, thus $\cT$ induces a clique partition $\cP$ of $G$. Note that computing this partition requires that we have an embedding of $G$ available as input; this is the only place where we use the embedding in this proof.

We can use a lemma by Ito and Kadoshita~\cite{ito-hamiltonian} to reduce the number of points in each tile to a constant that depends only on the maximum degree $\Delta$ of $G_\cP$; see also~\cite{ciac-hamiltonian}. Their key lemma is the following.

\begin{lemma}[Ito and Kadoshita \cite{ito-hamiltonian}]\label{lem:sampleedges}
Let $\cP$ be a partition of $\,V(G)$ into cliques where the maximum degree of $G_\cP$ is $\Delta$. Then for each pair of distinct cliques $C,C'\in \cP$, there is a way to remove all but $O(\Delta^2)$ edges among those with one endpoint in $C$ and the other in $C'$ so that the resulting graph $G_1$ has a Hamiltonian cycle if and only if $G$ has a Hamiltonian cycle.
\end{lemma}

As observed by~\cite{ciac-hamiltonian}, if $G_1$ is connected, removing the vertices from each clique of the clique partition in $G_1$ that do not have an edge to a vertex of any other clique in the partition preserves the Hamiltonicity of $G_1$. We thus obtain a reduced graph $G_2$, which contains at most $O(\Delta^3)$ vertices per tile. Let $G_3$ be the supergraph of $G_2$ where we reintroduce the deleted edges between the remaining vertices. Then $G_3$ is an induced subgraph of $G$, where in each clique of the partition there are $O(\Delta^3)$ vertices. Moreover, it has a Hamiltonian cycle if and only if $G_2$ does, since we can obtain $G_2$ from $G_3$ by applying Lemma~\ref{lem:sampleedges}. Therefore, $G_3$ is an induced subgraph of $G$ that has a Hamiltonian cycle if and only if $G$ does, and given $G$ and its partition into cliques, $G_3$ can be computed in polynomial time.


Clearly $G_3 \in \SNUBG_{\Hyp^d}(\rho,\nu,k)$ for some constant $k$ that depends only on $\rho,\nu$ and $d$, so by Theorem~\ref{thm:main} $G_3$ has treewidth $O(\log n)$ if $d=2$ or $O(n^{1-1/(d-1)})$ if $d\geq 3$. We can compute a corresponding tree decomposition in $\poly(n)$ (respectively, in $2^{O(n^{1-1/(d-1)})}$) time.

Finally, we run a treewidth-based Hamiltonian cycle algorithm~\cite{single-exponential} on this tree decomposition of $G_3$, which takes $2^{O(\tw(G_3))}\poly(n)$ time.
\end{proof}


Using the ($\cP$-flattened)-treewidth-based algorithms of~\cite{frameworkpaper} with the treewidth bounds of Theorem~\ref{thm:main}, we get the following results.

\begin{theorem} Let $\rho>0$, $\nu\geq 1$, $2\leq d\in \Nats$, and $k\in \Nats_+$ be fixed constants. Then
\textsc{Dominating Set, Vertex Cover, Feedback Vertex Set, Connected Dominating Set, Connected Vertex Cover} and \textsc{Connected Feedback Vertex Set} can be solved in
\begin{itemize}
	\item $\poly(n)$ time in $\SNUBG_{\Hyp^2}(\rho,\nu,k)$
	\item $n^{O(\log n)}$ time in $\NUBG_{\Hyp^2}(\rho,\nu)$
	\item $2^{O(n^{1-1/(d-1)})}$ time in $\NUBG_{\Hyp^d}(\rho,\nu)$ if $d\ge 3$.
\end{itemize}
\end{theorem}

\section{Lower bounds}\label{sec:lower}

The goal of this section is to prove Theorem~\ref{thm:lowerbound}, and to show that Euclidean lower bounds can be carried over to hyperbolic spaces of dimension at least 3.

\subsection{A lower bound for \IS in the hyperbolic plane}

Although the proof of Theorem~\ref{thm:lowerbound} is for a specific radius $\rho$, it will be apparent from the proof that it can be adapted to any other constant radius. Moreover, we state the theorem for $\nu=1$, and the lower bound then automatically follows for any $\NUBG_{\Hyp^2}(\rho,\nu)$ with $\nu >1$.

\begin{figure}[t]
\centering
\includegraphics[width=0.7\textwidth]{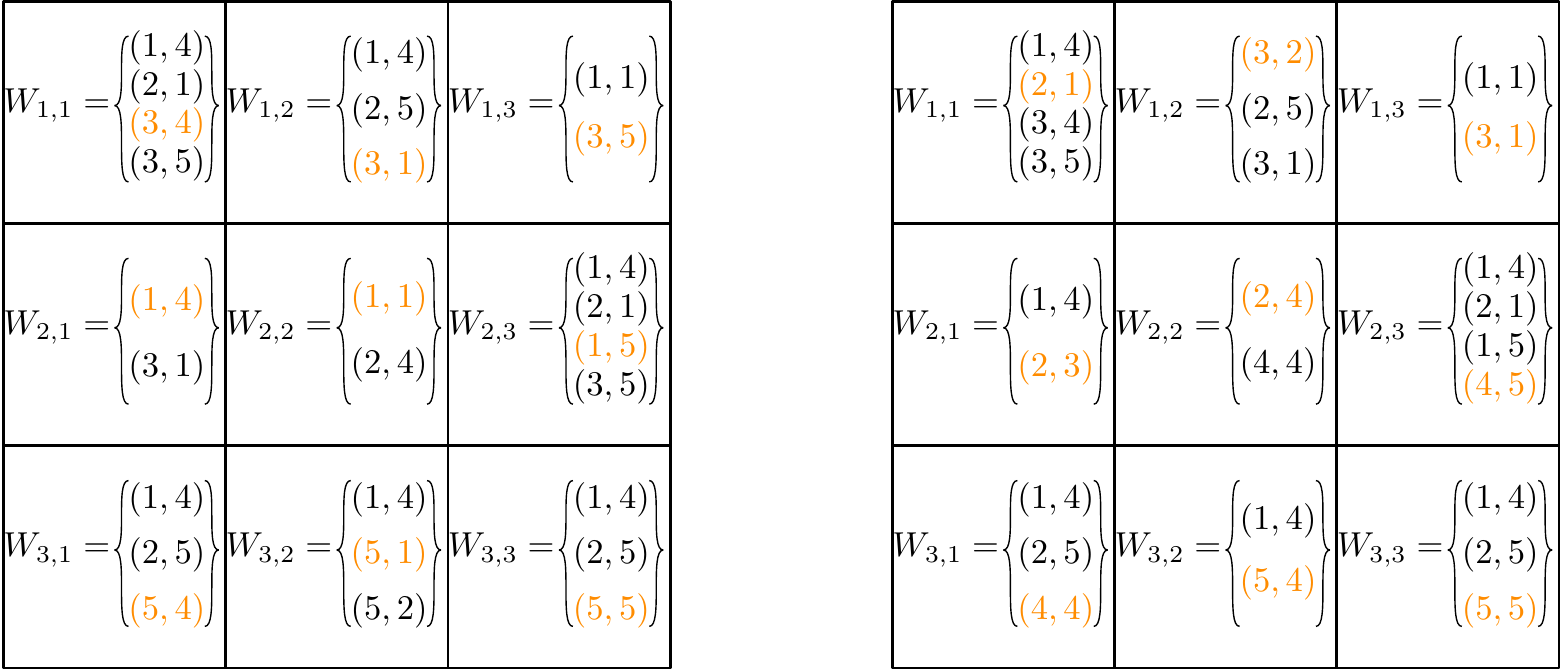}
\caption{In a \GT instance (left), one is given $k\times k$ sets $W_{a,b}\subseteq[n]\times[n]$, and the goal is to find a $w_{a,b} \in W_{a,b}$ for each $(a,b)\in [k]\times [k]$ so that horizontal/vertical neighbors share their first/second coordinate. In \GTleq (right) we require that these coordinates form non-decreasing sequences left-to-right and top-to-bottom. The examples have $k=3$ and $N=5$.}\label{fig:gt}
\end{figure}

The idea of the proof is to embed a blown-up grid structure in the hyperbolic plane, giving high multiplicity to the vertices of a plane tiling. The structure is then used to realize an instance of \GT~\cite{fptbook}. See Figure~\ref{fig:gt} for an illustration.
In an instance of \GT, we are given an integer $k$, an integer $n$, and a
collection $\cS$ of $k^2$ non-empty sets $W_{a,b} \subseteq \civ{n} \times
\civ{n}$ for $1 \leq a,b \leq k$. The goal is to decide if
there is a selection $w:[k]^2 \rightarrow [n]^2$ such that
$w_{a,b}\in W_{a,b}$ for each $1 \leq a,b \leq k$ and
\begin{itemize}
\setlength\itemsep{0em}
\item if $w_{a,b}=(x,y)$ and $w_{a+1,b}=(x',y')$, then $x=x'$;
\item if $w_{a,b}=(x,y)$ and $w_{a,b+1}=(x',y')$, then $y=y'$.
\end{itemize}

One can picture these sets in a $k\times k$ table: in each cell $(a,b)$, we
need to select an element from $W_{a,b}$ so that the
elements selected from horizontally neighboring cells agree in the
first coordinate, and those from vertically neighboring cells agree
in the second coordinate.
\GTleq is a variant of \GT where instead of equalities between neighboring cells, we only require inequalities:
\begin{itemize}
\setlength\itemsep{0em}
\item if $w_{a,b}=(x,y)$ and $w_{a+1,b}=(x',y')$, then $x\leq x'$;
\item if $w_{a,b}=(x,y)$ and $w_{a,b+1}=(x',y')$, then $y\leq y'$.
\end{itemize}

Our reduction to prove a lower bound for \IS in hyperbolic noisy uniform disk graphs has three steps: first, we reduce from a version of satisfiability to \GT, then invoke a reduction from \GT to \GTleq, and finally we reduce from \GTleq to \IS in noisy unit disk graphs.

\paragraph*{Step 1: from \niceSAT to \GT.}
A $(3,3)$-CNF formula is a boolean formula in conjunctive normal form which has at most three variables per clause, and each variable occurs at most three times: we call the decision problem of such formulas the $(3,3)$-SAT problem. By~\cite[Proposition 19]{frameworksArxiv} there is no $2^{o(n)}$ algorithm for \niceSAT under ETH, where $n$ is the number of variables. An instance of \GTn($n,\log n$) is a \GT instance where $k\leq \log n$ and for each $(a,b)\in [k]\times [k]$, we have $W_{a,b} \subseteq \civ{n} \times \civ{n}$.


It is known that both \GT and \GTleq has no $f(k)n^{o(k)}$ algorithm under ETH~\cite{fptbook}, and the following lemma seems to be a direct corollary of this for \GT($n,k$). Notice however that simply substituting $k=\log N$ is formally incorrect, as the known lower bound excludes the existence of an algorithm that works for \emph{all} $k$; it is still possible in principle that for all $k\leq \log n$ an $n^{o(k)}$ algorithm could be given. The next lemma excludes this possibility.

\begin{lemma}\label{lem:lowerS1}
There is no $2^{o(\log^2(n))}$ algorithm for \GTn$(n,\log n)$, unless ETH fails.
\end{lemma}

\begin{proof} Let $\phi$ be a \niceCNF formula of $n$ variables. In a
polynomial preprocessing step, we can eliminate clauses of size $1$ and
variables that occur only once; we henceforth assume that every variable in
$\phi$ occurs once or twice, and each clause has two or three literals. Notice
that $\phi$ has $m<2n$ clauses. We group these clauses into $\sqrt{m}$ groups
$G_1,\dots,G_{\sqrt{m}}$, each of size $\sqrt{m}$. In each group $G_a$ there are at
most $3\sqrt{m}$ variables that occur; let us denote these variables by $V_a$. These variables have at most
$N\eqdef 2^{3\sqrt{m}}$ possible truth assignments. We enumerate all possible
assignments for $V_a$, and index them from $1$ to at most $N$, getting the assignments $V_a(1), V_a(2),\dots$.

We create an instance of \GT where we index the rows and columns of the grid
with these groups of clauses, so there are $k\eqdef\sqrt{m}\leq \log N$ rows and
columns. For each pair $a,b$ with $1\leq a,b \leq k$, we define the set $W_{a,b}$ as follows.
\smallskip
The set $W_{a,b}$ contains a pair $(x,y)\in [N]\times[N]$ if and only
if
\begin{itemize}
\item[(i)] assignments of index $x$ and $y$ are consistent, i.e., for any variable
that occurs in both $V_a$ and $V_b$ the same truth value is given in assignment $V_a(x)$ as in assignment $V_b(y)$, and
\item[(ii)] all clauses in $G_a$ and $G_b$ are
satisfied by setting $V_a$ according to $V_a(x)$ and $V_b$ according to $V_b(y)$.
\end{itemize}

The instance created this way has at most $N\cdot
N=2^{O(\sqrt{m})}=2^{O(\sqrt{n})}$ entries in each tile, and altogether $m$
tiles, so the size is $2^{O(\sqrt{n})}\cdot m = 2^{O(\sqrt{n})}$.

If $\phi$ is satisfiable, then the satisfying assignment defines a solution to
this \GT instance. If the \GT instance has a solution, then the union of the
assignments selected in each tile is consistent (that is, each variable
receives a unique value) by property (i). Moreover, in each group the clauses
are satisfied, therefore the formula is satisfied. Hence, the \GT instance has
a solution if and only if $\phi$ is satisfiable.
Given $\phi$, the \GT instance can be created in $2^{O(\sqrt{n})}$ time.
Suppose now that there is a $2^{o(\log^2(N))}$ algorithm for \GTn($N,\log N$).
Applying such an algorithm to the instance created above, we would get an
algorithm for \niceSAT with running time
\[2^{O(\sqrt{n})} + 2^{o(\log^2(N))} = 2^{O(\sqrt{n})} + 2^{o((c\sqrt{n})^2)}=2^{o(n)},\]
which contradicts ETH.
\end{proof}

\paragraph*{Step 2: from \GT to \GTleq.}
\begin{lemma}\label{lem:lowerS2}
There is no $2^{o(\log^2(n))}=n^{o(\log n)}$ algorithm for \GTleq with parameters $(n,\log n)$, unless ETH fails.
\end{lemma}

\begin{proof}
The Lemma follows from Lemma~\ref{lem:lowerS1} and the reduction of Cygan~\etal~\cite[Theorem 14.30]{fptbook}, which takes an instance of \GTn($n,k$) and in polynomial time creates an equivalent instance of \GTleq with parameters $(3n^2(k+1)+n^2+n,4k)$.
\end{proof}

\paragraph*{Step 3: from \GTleq to \IS in $\NUBG_{\Hyp^2}(\rho,1)$.}
Given Lemma~\ref{lem:lowerS2}, all we need now to prove Theorem~\ref{thm:lowerbound} is the following.


\begin{restatable}{lemma}{lemconstruction}\label{lem:construction}
Given an instance $\cI$ of \GTleq with parameters $(n,\log n)$, we can create in $\poly(n)$ time a graph $H\in \NUBG_{\Hyp^2}(\rho,1)$ on $\poly(n)$ vertices which has an independent set of a certain size $k$ if and only if $\,\cI$ is a yes-instance.
\end{restatable}

\begin{proof}
The construction is built on the uniform tiling $\cT$ with regular pentagon tiles, where at each vertex four pentagons meet. (This is the uniform tiling with Schl\"afli symbol $(5,4)$.) Let $\delta$ be the diameter of the tiles. Let $o$ be the origin, chosen at the center of a pentagon. The side length of the pentagons is exactly $\rho\eqdef\cosh^{-1}(\cos(\pi/5)/\sin(\pi/4))$; see the proof of Lemma~\ref{lem:tilings}(i) for a similar calculation. Let $\Gamma$ be the infinite plane graph defined by the vertices and edges of the tiling $\cT$. See Figure~\ref{fig:gridembed}.

\begin{figure}[t]
\centering
\includegraphics[width=.7\textwidth]{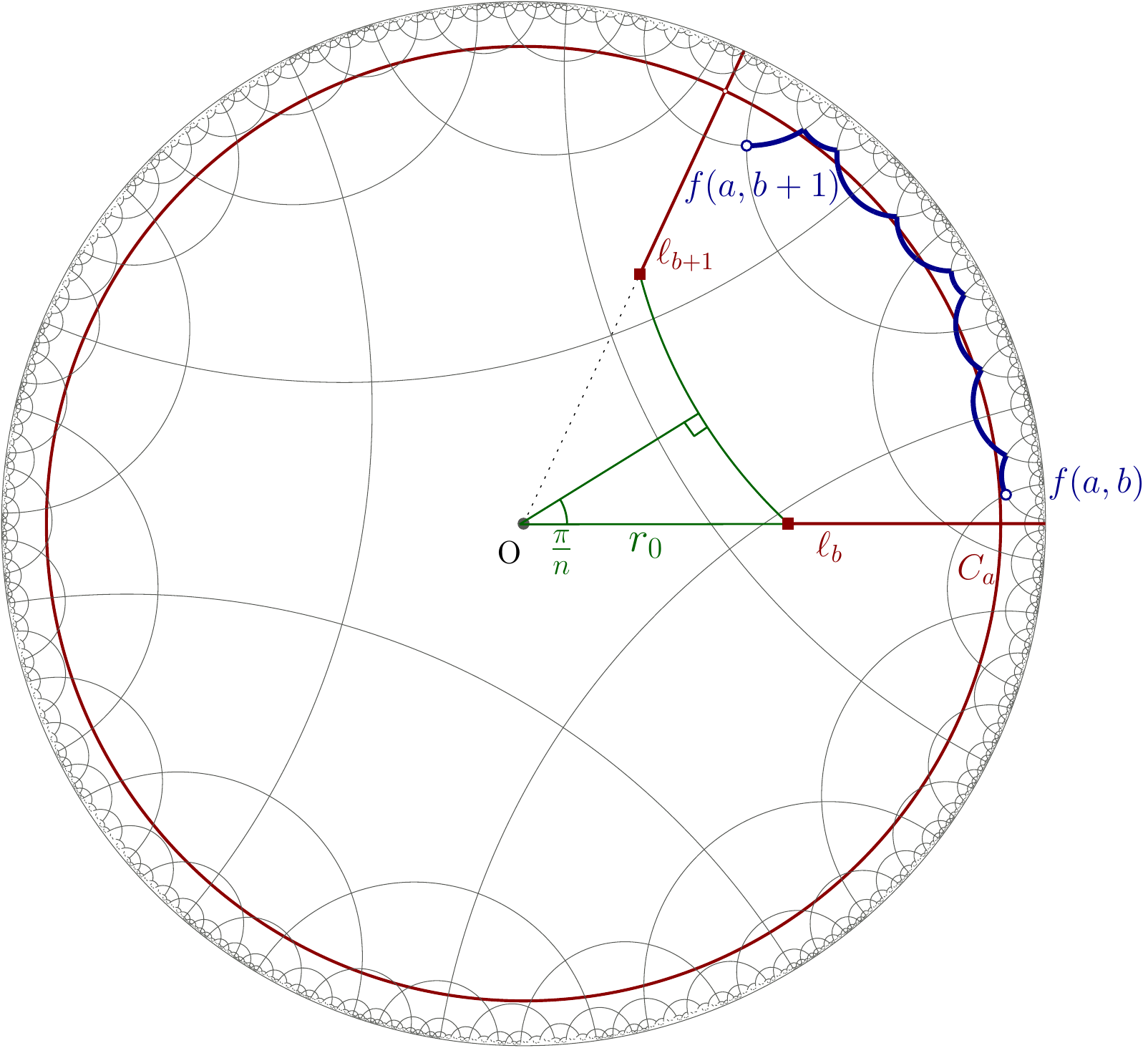}
\caption{Finding vertices for the grid points $(a,b)$ and $(a,b+1)$, together with a canonical path between them.}\label{fig:gridembed}
\end{figure}

\medskip
\claim{
The graph $\Gamma$ contains a subgraph $G$ of size $\poly(n)$ that is a subdivision of the $\log n\times\log n$ grid.
}
{
Let $\beta\eqdef 2\pi/n$, and let $r_0\eqdef \sinh^{-1} (\sinh(4\delta) n)=O(\log n)$. 

From this point onward we use polar coordinates $(r,\phi)$ around $o$. For
each $a\in [\log n]$ let $C_a$ be the circle around the origin with
radius $r_0+9\delta \cdot a$ (see Figure~\ref{fig:gridembed}). For each $b\in
[\log n]$ let $\ell_b$ be the radial half-line with equation $\phi=b\cdot
\beta$ whose starting point is on $C_1$ (i.e., $o\not\in \ell_b$). Let
$S_a\subset \cT$ be the set of tiles intersected by $C_a\, (a\geq 0)$. Note
that $\Nei_{\Hyp^2}(C_a,4\delta)\cap
\Nei_{\Hyp^2}(C_{a+1},4\delta)=\emptyset$. Similarly,  we define the the tile
set $S'_b$ as the set of tiles intersected by $\ell_b$, and claim that
$\Nei_{\Hyp^2}(\ell_b,4\delta)\cap
\Nei_{\Hyp^2}(\ell_{b+1},4\delta)=\emptyset$. To see this, it is sufficient to
show that $\dist(\ell_b,\ell_{b+1})> 8\delta$. Note that the distance of
$\ell_b$ and $\ell_{b+1}$ is realized at their starting points, so it can be
calculated exactly using the right-angle triangle given by the line through
$\ell_b$, the angle bisector of $\ell_b$ and $\ell_{b+1}$, and the line
through the starting points of $\ell_{b}$ and $\ell_{b+1}$, depicted in
green in Figure~\ref{fig:gridembed}. The triangle has angle $\pi/n$ at $o$,
and its hypotenuse has length $r_0$.
\begin{align*}
\dist_{\Hyp^2}(\ell_b,\ell_{b+1})&=2\sinh^{-1}\left(\sin \frac{\pi}{n}\sinh r_0\right)\\
&> 2\sinh^{-1}\left(\frac{\pi}{2n}\sinh(4\delta) n\right)\\
&>8\delta.
\end{align*}

For each  $(a,b)\in [\log n]\times[\log n]$, let $f(a,b)$ be an arbitrary vertex of the tile containing the intersection of $C_a$ and $\ell_b$, i.e., the tile containing the point $(r_0+9\delta\cdot a, \beta\cdot b)$.  (By introducing small perturbations to $\beta$ and $r$, we can make sure that these intersections are all interior points of some tile.)
The above observations about the distance of neighboring circles $C_a$ and half-lines $\ell_b$ imply that each tile contains at most one point from the image of $f$.

We can connect $f(a,b)$ and $f(a,b+1)$ in $\Gamma$ using only vertices and edges of the tile set $S_a$. Similarly, we can connect $f(a,b)$ and $f(a+1,b)$ using only vertices and edges of the tile set $S'_b$. We create such paths for all neighboring pairs of the grid $[\log n]\times[\log n]$. Note that paths created this way may have overlapping inner vertices, but we can avoid such overlaps by doing local modifications in the $4\delta$-neighborhoods of the tiles containing $f(a,b)$ for each $(a,b)\in [\log n]\times[\log n]$. The overlaps can occur only if  the vertex $v\in V(\Gamma)$ is incident to a tile $T\in S_a$ and a (not necessarily distinct) tile $T'\in S'_b$; but this can only happen if $v \in \Nei_{\Hyp^2}(C_a, \delta)\cap \Nei_{\Hyp^2}(\ell_b, \delta)$, therefore $v\in \Nei_{\Hyp^2}(C_a\cap\ell_b, 3\delta) \subseteq B(f(a,b),4\delta)$, so indeed it is sufficient to create a subdivision of a $4$-star in the intersection of $\Gamma$ with the balls $B(f(a,b),4\delta)$ with given boundary vertices. After the local modifications, we are left with paths whose interiors are vertex disjoint. We call these paths \emph{canonical paths}.

These canonical paths give a subgraph of $\Gamma$ that is a subdivision of the $\log n\times\log n$ grid. Notice that all canonical paths stay within a radius $r_0+9\delta\cdot \log n+4\delta=O(\log n)$ ball around the origin; this ball contains $\poly(n)$ tiles by Proposition~\ref{prop:tilesballs} (i), and therefore the number of vertices in this ball (and in the subgraph) is $\poly(n)$. 
}\index{canonical!t@\tilde path}
\medskip

Now consider an instance of \GTleq with parameters $(n,\log n)$ and sets $W_{a,b}\; (1\leq a,b \leq \log n)$. Using the claim above, we can construct a graph $H\in \NUBG_{\Hyp^2}(\rho,1)$ with the desired properties, as explained next.
Let $G$ be the subdivision of the $[\log n]\times[\log n]$ grid provided by the claim above. Vertices of $G$ are identified with the corresponding point in $\Hyp^2$. At each vertex of $G$, we place \emph{at most} $n^2$ copies of the vertex, and we index these copies by some specified subset of $[n]\times[n]$, as explained later. The copies will form a multiset embedding of the vertex set of the noisy unit disk graph $H$ that we are constructing. This will result in a noisy unit disk graph  where points placed in the same vertex of $G$ are all connected, and any pair of points further away than $2\rho$ are never connected; for points at exactly $2\rho$ distance away (i.e., at neighboring vertices of $G$), they are either connected or not connected. The size is $|V(H)|\leq n^2|V(G)|=\poly(n)$.

Let $f(a,b)$ denote the vertex of $G$ corresponding to the grid point $(a,b)$, and let $\mypath(f(a,b), f(a+1,b))$ denote the set of vertices in the canonical path between $f(a,b)$ and $f(a+1,b)$. Similarly, $\mypath(f(a,b), f(a,b+1))$ denotes the set of vertices in the canonical path between $f(a,b)$ and $f(a,b+1)$. We refer to a point of $H$ of index $(x,y)\in [n]\times[n]$ at vertex $v$ of $G$ as $v(x,y)$. There is a natural orientation of $G$ where we orient the edge from $u$ to $v$ if $u,v$ occurs in this order on the canonical $G$-path $\mypath(f(a,b), f(a+1,b))$ or $\mypath(f(a,b), f(a,b+1))$. 

We now define the vertices of $H$. For each vertex $v\in V(G)$, we define a subset of valid indices from $[n]\times[n]$, each of which defines a point in the embedding of $H$.
For each $v=f(a,b)$ where $(a,b)\in [\log n]\times[\log n]$, we create a point $v(x,y)$ if and only if $(x,y)\in W_{a,b}$. If $v$ is an internal vertex of $\mypath(f(a,b),f(a+1,b))$, then we only create the points $v(x,1)$ for each $x\in [n]$, and if $v$ is an internal vertex of $\mypath(f(a,b), f(a,b+1))$, then we only create the points $v(1,y)$ for each $y\in [n]$.
The edges of $H$ are defined as follows. Let $uv$ be an arc of $G$.
If $(uv)\in \mypath(f(a,b), f(a+1,b))$, then connect $u(x,y)$ to $v(x',y')$ if and only if $x > x'$.
If $(uv)\in \mypath(f(a,b), f(a,b+1))$, then connect $u(x,y)$ to $v(x',y')$ if and only if $y > y'$. This finishes the construction of $H$. It remains to show that $H$ has an independent set of size $|V(G)|$ if and only if the original \GTleq
 instance has a solution.

\medskip
Consider a solution $s:[\log n]\times[\log n]\rightarrow [n]\times[n]$ of
the instance of \GTleq. Let us use the notation
$s(a,b)\eqdef(x^s_{a,b},y^s_{a,b})$. We can construct an independent set $I$
in $H$ of size $|V(G)|$ based on $s$ the following way. For each $v=f(a,b)$,
add $v(x^s_{a,b},y^s_{a,b})$ to $I$. If $v$ is an internal vertex of
$\mypath(f(a,b), f(a+1,b))$, we put $v(x^s_{a,b},1)$ into $I$. If $v$ is an
internal vertex of $\mypath(f(a,b), f(a,b+1))$, then we put $v(1,y^s_{a,b})$
into $I$. By our definition of $H$, edges only occur between vertices of $H$
that correspond to neighboring vertices on a canonical path. For each arc of a canonical path, we have selected the same first/second index vertex from $H$ at the source and the target, except at the end of the path, where we selected a
first/second index that is greater or equal than the first/second
index at the source of the arc (due to the \GTleq solution). Therefore, none of these vertex pairs form an edge in $H$, and $I$ is an independent set.

Conversely, suppose that $H$ has an independent set $I$ of size $|V(G)|$.
Since points of $H$ corresponding to any given vertex in $G$ form a clique,
$I$ has to contain exactly one point $v(x,y)$ for each $v\in V(G)$. 
For each $(a,b)\in [\log n]\times[\log n]$ let
$w_{a,b} \in [n]\times [n]$ be the index of the vertex in $I$ that is assigned by the embedding to $f(a,b)$.
We claim that $s(a,b)\eqdef w_{a,b}\; ((a,b)\in
[\log n]\times[\log n])$ is a solution to the \GTleq instance. To see this,
consider a canonical path $\mypath(f(a,b), f(a+1,b))$. Let $uv$ be an arc
along this path. By the definition of the edges of $H$, if $u(x,y)\in I$ and
$v(x',y')\in I$, then $x\leq x'$. Therefore, the first coordinate
of $w_{a,b}$ is less or equal to the first coordinate of $w_{a+1,b}$. An
analogous argument shows that the inequality is carried correctly also from
$w_{a,b}$ to $w_{a,b+1}$. Therefore, $H$ has an independent set of size
$k=|V(G)|$ if and only if the \GTleq instance has a solution.
\end{proof}

\begin{remark}
The above proof uses a specific radius $\rho$, but one could adapt it to any other constant radius. First, notice that we could use a different regular tiling as basis, as long as the tiling is stays $4$-regular. Similarly to our tilings in Lemma~\ref{lem:tilings}(i), we can use a regular tiling of Schl\"afli symbol $(2^{\delta+2},4)$ for any $\delta\in \Nats_+$ and get a tiling of diameter $\Theta{\delta}$. Finally, notice that we strictly followed the vertices of $\Gamma$ when creating the grid subdivision $G$, but this is not necessary in general. One can create a construction for any radius $c\rho$ for $c\in [\eps,1/4]$ (for any fixed $\eps>0$) by replacing each edge $uv$ of $G$ with a carefully chosen path $u=x_1,x_2\dots,x_k=v$ with $k=\ceil{c/x}$, where the distance of neighbors is exactly $c\rho$, and the distance of non-neighbors is strictly larger than $c\rho$. By picking the appropriate tiling and a subdivision, we can create a lower bound for any constant radius. 
\end{remark}

\subsection{Higher-dimensional lower bounds}\label{sec:app_lower_highdim}

The key to the above construction was to embed a grid into $\Hyp^2$. If $d\ge
3$, it is much easier to embed a $(d-1)$-dimensional grid $[n]^{d-1}$ into
$\Hyp^d$, as explained next.

Consider the
half-space model $\cH = \{(x_1,\dots,x_d)\in \Reals^d \mid x_d>0\}$ of $\Hyp^d$.
A \emph{horosphere}\index{horosphere} of $\Hyp^d$ is either a Euclidean sphere in $\cH$ that touches the boundary $\bd \cH$ (without the boundary point), or a hyperplane in $\cH$ parallel to $\bd \cH$.
It is well-known that any horosphere $\cS$ in $\Hyp^d$ is isometric to
$\Reals^{d-1}$; but this is true only with respect to the Riemannian
metric of $\cS$ inherited from $\Hyp^d$~\cite{benedetti2012lectures}. This
inherited metric of $\cS$ is denoted by $\dist_\cS$. We claim that for any
pair of points $x,y\in \cS$ we have $\dist_\cS(x,y)\leq1$ if and only if
$\dist_{\Hyp^d}(x,y)\leq c$ for some constant $c$. Fix a value $t>0$ and let $\cS$ be a horosphere that is a Euclidean hyperplane
parallel to $\bd\cH$, that is, $\cS\eqdef \{(x_1,\dots,x_d)\in \cH \mid x_d=t\}$. Now $\dist_\cS$
is just a scaled Euclidean metric, i.e., for any $x,y\in \cS$, we have
$\dist_\cS(x,y)=c_t\norm{x-y}$, where $\norm{.}$ is the Euclidean norm and $c_t$ is a constant that depends only on $t$. The
hyperbolic distance between two points $x,y \in \cS$ is
\[\dist_{\Hyp^2}(x,y) =
2\sinh^{-1}\left(\frac{\norm{x-y}}{2t}\right).\]

It follows that for $x,y \in \cS$, we have
$\dist_{\Hyp^2}(x,y)= 2\sinh^{-1}\left(\frac{\dist_\cS(x,y)}{2c_tt}\right)$.
Hence, the two distance functions are mapped to each other with the monotone increasing
bijective function $f:\Reals_{\geq 0}\rightarrow \Reals_{\geq 0}\,, f(z) =
2\sinh^{-1}\left(\frac{z}{2c_tt}\right)$ that depends only on the choice of
$t$.
Therefore, any induced grid graph in $\Ints^{d-1}$ can be realized as a uniform  ball graph $\UBG_{\Hyp^d}(c'_t)$ for some constant $c'_t$, or even more generally, any unit ball graph of $\Reals^{d-1}$ is a uniform ball graph in $\Hyp^d$, where the radius $\rho$ is a constant independent of $n$.
In particular, all lower bounds discussed in \cite{frameworkpaper,frameworksArxiv,DBLP:journals/jocg/BiroBMMR18} hold in unit ball graphs. Consequently, if $d\ge 3$, then all of our algorithms have a lower bound of $2^{\Omega(n^{1-1/(d-1)})}$, i.e., the algorithms given in Section~\ref{sec:algs} have optimal exponents up to constant factors in the exponent under ETH.

\section{Conclusion}

We have established that shallow noisy uniform ball graphs in $\Hyp^2$ have treewidth $O(\log n)$, while their non-shallow counterparts have $\cP$-flattened treewidth $O(\log^2 n)$ for any greedy partition $\cP$. For higher dimensions, we have established a bound of $O(n^{1-1/(d-1)})$ on the $\cP$-flattened treewidth for greedy partitions. These bounds implied polynomial, quasi-polynomial and subexponential algorithms in the respective graph classes, where the exponents in $\Hyp^d$ match those in $\Reals^{d-1}$ with the exception of $\Hyp^2$. We have established that the quasi-polynomial algorithm given for \IS in $\NUBG_{\Hyp^2}$ is optimal up to constant factors in the exponent under ETH. We would like to call the attention to three interesting directions for future research.

\begin{itemize}
\item \textbf{Generalizing the underlying space} 
Our tools exhibit a lot of flexibility, and likely can go beyond standard hyperbolic space. As seen in the paper by Krauthgamer and Lee~\cite{KrauthgamerL06}, there are very powerful tools available in $\delta$-hyperbolic spaces; it would be interesting to see ($\cP$-flattened) treewidth bounds in this more general setting.
\item \textbf{Specializing the graph classes and algorithms} 
Our polynomial and quasi-polynomial algorithms have large constants in the exponents of their running times. For the special case of hyperbolic grid graphs (finite subgraphs of the graph of a fixed regular tiling of $\Hyp^2$), it should be possible to improve these running times. Is there a different algorithm with more elementary tools with a similar or better running time? How is the type of the tiling represented in the optimal exponent? What lower bound tools could be used here?
\item \textbf{Recognition} Our algorithms are not capable of deciding if a graph given as input is in a particular class $\NUBG_{\Hyp^d}(\rho,\nu)$. This question is probably easier to attack for hyperbolic grid graphs. Although Euclidean grid graphs are \NP-hard to recognize~\cite{DBLP:journals/tcs/SaFMF11}, hyperbolic grid graphs seem more tractable. Can hyperbolic grid graphs be recognized in polynomial time?
\end{itemize}

\paragraph*{Acknowledgments}
The author thanks Mark de Berg, Hans L. Bodlaender, and G\"unter Rote for their comments, and Hsien-Chih Chang for a discussion about this work.


\clearpage
\appendix

\section{A note on the number of (noisy) uniform ball graphs}\label{sec:app_nubg}

\begin{theorem}\label{thm:noisybig}
Let $d \ge 2$ be a fixed constant. Then the set of graphs on $n$ vertices in
\begin{enumerate}
	\item[(i)] $\UBG_{\Hyp^d}(\rho)$ has size $2^{O(n\log n)}$
	\item[(ii)] $\SNUBG_{\Hyp^d}(\rho, \nu,k)$ has size $2^{O(n\log n)}$
	\item[(iii)] $\NUBG_{\Hyp^d}(\rho, \nu)$ has size $2^{\Theta(n^2)}$
\end{enumerate}
\end{theorem}

\begin{proof}
(i) Clearly any graph defined by a set of balls in $\Hyp^d$ is also a ball graph in $\Reals^d$ due to the Poincar\'e ball model, therefore an upper bound on the number of ball graphs in $\Reals^d$ suffices. This can be proved using Warren's theorem~\cite{MCDIARMID2014413}.

(ii) We denote the closed ball around $p\in \cM$ of radius $r$ by $B_\cM(p,r)$. Note that $\SNUBG_{\Hyp^d}(\rho,\nu,k)$ has maximum degree $O(1)$, since the points representing the neighborhood of a vertex $v$ can be covered by $B(\eta(v),2\rho\nu)$, which in turn can be covered by $O(1)$ balls of radius $\rho$, each of which can contain no more than $k$ points. Therefore an edge list representation of a graph on $n$ vertices has $O(n)$ symbols from an alphabet of size $n$; this can distinguish no more than $2^{O(n\log n)}$ graphs.

(iii) The upper bound follows from the fact that there are $2^{\binom{n}{2}}$ labeled graphs, which is an upper bound on the number of graphs. For the lower bound, we can realize any co-bipartite graph (a graph that is the complement of a bipartite graph) in $\NUBG_\cM(\rho,1)$, where $\cM$ is an arbitrary metric space that contains a pair of points at distance $2\rho$ (that is, one can choose $\cM=\Hyp^d$).
Consider co-bipartite graphs that have an even split into clique $A$ on $\ceil{n/2}$ and clique $B$ on $\floor{n/2}$ vertices. For a given vertex $a\in A$, there are $2^{\floor{n/2}}$ neighborhoods to choose from; we can choose for each vertex $a\in A$ independently, which gives $(2^{\floor{n/2}})^{\ceil{n/2}}=2^{\Omega(n^2)}$ possible choices, each resulting in different labeled graphs. Each unlabeled graph has been counted at most $n!=2^{\Theta(n\log n)}$ times, so there are $2^{\Omega(n^2)}$ distinct unlabeled co-bipartite graphs.
\end{proof}

\clearpage

\bibliography{hyperbolic}

\end{document}